%% file: fast-katz-im.tex
\documentclass[10pt]{article}

\input{fast-katz-setup}

\input{fast-katz-custom}
\ifpdf%
  \RequirePackage[colorlinks,pdfdisplaydoctitle
      ,citecolor=theblue
      ,linkcolor=theblue
      ,urlcolor=theblue
      ,hyperfootnotes=false,pdftex,pagebackref]{hyperref}%
\else
  \RequirePackage[colorlinks,pdfdisplaydoctitle
      ,citecolor=theblue
      ,linkcolor=theblue
      ,urlcolor=theblue
      ,hyperfootnotes=false,dvipdfm,pagebackref]{hyperref}
\fi
\renewcommand*{\backref}[1]{}
\renewcommand*{\backrefalt}[4]{%
  \ifcase #1 %
    No citations.%
  \or
    Cited on page #2.%
  \else
    Cited on pages #2.%
  \fi
}

\usepackage{algorithm}
\usepackage[noend]{algorithmic}

\title{Fast matrix computations for pair-wise and column-wise\\
commute times and Katz scores}
\author{
Francesco Bonchi\\
Yahoo! Research\\ 
Barcelona, Spain\\
\email{bonchi@yahoo-inc.com}
\and 
Pooya Esfandiar\\
Univ. of British Columbia\\
Vancouver, BC \\
\email{pooyae@cs.ubc.ca}
\and 
David F.~Gleich\\
Sandia National\\ Laboratories\\
Livermore, CA\\
\email{dfgleic@sandia.gov}
 \and 
Chen Greif\\
Univ. of British Columbia\\
Vancouver, BC\\
\email{greif@cs.ubc.ca}
\and 
Laks V.~S.~Lakshmanan\\
Univ. of British Columbia\\
Vancouver, BC\\
\email{laks@cs.ubc.ca}
}

\begin{document}

\maketitle

\input{sec-abstract}
\input{sec-intro}
\input{sec-related}
\input{sec-lanczos}
\input{sec-pairwise}

\input{sec-topk}

\input{sec-experiments}
\input{sec-conclusion}

\bibliographystyle{plainnat}
{\RaggedRight 
\bibliography{sigmod09}
}

\let\thefootnote\relax
\footnotetext{
       Sandia National Laboratories is a multi-program laboratory
       managed and operated by Sandia Corporation, a wholly owned
       subsidiary of Lockheed Martin Corporation, for the U.S.
       Department of Energy's National Nuclear Security Administration
       under contract DE-AC04-94AL85000.}
 
\end{document}

%% file: fast-katz-setup.tex
\usepackage{ragged2e}
\usepackage{geometry}
\makeatletter
\renewcommand\normalsize{%
   \@setfontsize\normalsize\@xpt{14}%
   \abovedisplayskip 10\p@ \@plus2\p@ \@minus5\p@
   \abovedisplayshortskip \z@ \@plus3\p@
   \belowdisplayshortskip 6\p@ \@plus3\p@ \@minus3\p@
   \belowdisplayskip \abovedisplayskip
   \let\@listi\@listI}
\normalsize
\renewcommand\small{%
   \@setfontsize\small\@ixpt{12}%
   \abovedisplayskip 8.5\p@ \@plus3\p@ \@minus4\p@
   \abovedisplayshortskip \z@ \@plus2\p@
   \belowdisplayshortskip 4\p@ \@plus2\p@ \@minus2\p@
   \def\@listi{\leftmargin\leftmargini
               \topsep 4\p@ \@plus2\p@ \@minus2\p@
               \parsep 2\p@ \@plus\p@ \@minus\p@
               \itemsep \parsep}%
   \belowdisplayskip \abovedisplayskip
}
\renewcommand\footnotesize{%
   \@setfontsize\footnotesize\@viiipt{10}%
   \abovedisplayskip 6\p@ \@plus2\p@ \@minus4\p@
   \abovedisplayshortskip \z@ \@plus\p@
   \belowdisplayshortskip 3\p@ \@plus\p@ \@minus2\p@
   \def\@listi{\leftmargin\leftmargini
               \topsep 3\p@ \@plus\p@ \@minus\p@
               \parsep 2\p@ \@plus\p@ \@minus\p@
               \itemsep \parsep}%
   \belowdisplayskip \abovedisplayskip
}
\renewcommand\scriptsize{\@setfontsize\scriptsize\@viipt\@viiipt}
\renewcommand\tiny{\@setfontsize\tiny\@vpt\@vipt}
\renewcommand\large{\@setfontsize\large\@xipt{15}}
\renewcommand\Large{\@setfontsize\Large\@xiipt{16}}
\renewcommand\LARGE{\@setfontsize\LARGE\@xivpt{18}}
\renewcommand\huge{\@setfontsize\huge\@xxpt{30}}
\renewcommand\Huge{\@setfontsize\Huge{24}{36}}
\makeatother

\usepackage[tableposition=top]{caption}

%% file: fast-katz-custom.tex
\usepackage{ifpdf}
\usepackage{xcolor}
\usepackage{paralist}
\usepackage{tabularx}
\usepackage{booktabs}
\usepackage{graphicx}
\usepackage{natbib}
\renewcommand{\cite}{\citep}
\usepackage{amssymb}
\usepackage{amsmath}

\usepackage{caption}
\usepackage{subfig}
\usepackage{array}
\usepackage{url}
\usepackage{float}

\ifpdf
  \usepackage{epstopdf}
	\usepackage{microtype}
\fi
\graphicspath{{./}{./figures/}}

\colorlet{TufteRed}{red!80!black}

\definecolor{theblue}{RGB}{0,0,180}
\colorlet{thered}{TufteRed}

\newcommand{\email}[1]{\textsf{\small #1}}

\usepackage{dgleich-math}
\renewcommand{\mat}{\boldsymbol}
\newcommand{\meye}{\mI}

\newcommand{\mLhat}{\mat{\tilde{L}}}
\DeclareMathOperator{\vol}{vol}

\newcommand{\blower}{\underline{b}}
\newcommand{\bupper}{\overline{b}}
\newcommand{\lmin}{\underline{\lambda}}
\newcommand{\lmax}{\overline{\lambda}}
\newcommand{\rhomin}{\underline{\rho}}
\newcommand{\rhomax}{\overline{\rho}}
\newcommand{\kappamin}{\underline{\kappa}}
\newcommand{\kappamax}{\overline{\kappa}}
\newcommand{\dmax}{d_{\text{max}}}

\newcommand{\dupper}{\overline{d}}
\newcommand{\dlower}{\underline{d}}

\newcommand{\omegaupper}{\overline{\omega}}
\newcommand{\omegalower}{\underline{\omega}}
\newcommand{\lasti}{_{-1}}

\newcommand{\gupper}{\overline{g}}
\newcommand{\glower}{\underline{g}}
\newcommand{\hupper}{\overline{h}}
\newcommand{\hlower}{\underline{h}}

%% file: sec-abstract.tex
\begin{abstract}
We first explore methods for approximating the commute time
and Katz score between a pair of nodes.  These methods are based on the
approach of matrices, moments, and quadrature
developed in the numerical linear algebra community.
They rely on the Lanczos process
and provide upper and lower bounds on an estimate
of the pair-wise scores. We also explore methods to
approximate the commute times and Katz scores
from a node to all other nodes in the graph.  Here,
our approach for the commute times is based on
a variation of the conjugate gradient algorithm, and it provides an estimate of all the 
diagonals of the
inverse of a matrix. Our technique for the Katz
scores is based on exploiting an empirical localization
property of the Katz matrix.  We adopt algorithms used
for personalized PageRank computing to these Katz scores
and theoretically show that this approach is convergent.
We evaluate these methods on 17 real world graphs
ranging in size from 1000 to 1,000,000 nodes.
Our results show that our pair-wise commute time method
and column-wise Katz algorithm both have attractive
theoretical properties and empirical performance.
\end{abstract}

%% file: sec-intro.tex
\section{Introduction} \label{sec:intro}

Commute times~\cite{Gobel1974-random-walks} and
Katz scores~\cite{katz} are two topological measures
defined between any pair of vertices in a graph that
capture their relationship due to the link
structure.  Both of these
measures have become important
because of the their use in
social network analysis as well as applications
such as link prediction~\cite{cikm03},
anomalous link detection~\cite{rattigan05},
recommendation~\cite{sarkar-moore07}, and
clustering~\cite{saerens}.

For example, \citet{cikm03} identify a variety of topological measures as features for
link prediction: the problem of predicting the likelihood of users/entities forming
new connections in the future, given the current state of the network. The measures they
studied fall into two categories -- neighborhood-based measures and path-based measures.
The former are cheaper to compute, yet the latter are more effective at link
prediction. Katz scores were among the most effective path-based measures studied
by \citet{cikm03}, and the commute time also
performed well.

Standard algorithms to compute these measures between all
pairs of nodes are often based on direct solution methods and 
require cubic time and quadratic space in the number
of nodes of the graph.  Such algorithms are impractical for modern
networks with millions of vertices and tens of millions of edges.
We explore algorithms to compute a targeted subset of scores that
do scale to modern networks.
%In what follows, we formally define these measure,
%and simultaneously both the contributions and overall structure
%of the paper.

Katz scores measure the affinity between nodes via
a weighted sum of the number of paths between them.
Formally, the Katz score between node $i$ and $j$ is
\[
K_{i,j} = \sum_{\ell=1}^{\infty} \alpha^\ell \mathrm{paths}_\ell(x,y) ,
\]
where $\text{paths}_\ell(x,y)$ denotes the number
of paths of length $\ell$ between $i$ to $j$ and $\alpha < 1$ is an
attenuation parameter.  Now, let $\mA$ be the symmetric adjacency matrix,
corresponding to a undirected and connected graph,
and recall that $(\mA^\ell)_{i,j}$ is the number of paths
between node $i$ and $j$.  Then computing the Katz scores
for all pairs of nodes is equivalent to the following computation:
\[ \label{eq:katz_matrix}
\mK = \alpha \mA + \alpha^2 \mA^2
+ \cdots = (\mI - \alpha \mA)^{-1} - \mI. \]
Herein, we refer to $\mK$ as the \emph{Katz} matrix.
We shall only study this problem when $\eye - \alpha \mA$
is positive definite, which occurs when 
$\alpha < 1/\sigma_{\max}(\mA)$ and also corresponds
to the case where the series expansion converges.  

In order to define the commute time between nodes, we
must first define the hitting time between nodes.
Formally, the commute time between nodes is defined as the sum of
hitting times from $i$ to $j$ and from $j$ to $i$,
and the hitting time from node $i$ to $j$ is the
expected number of steps for a random walk
started at $i$ to visit $j$ for the first time.
The hitting time is computed via first-transition
analysis on the random walk transition matrix associated
with a graph.  To be precise, let $\mA$, again,
be the symmetric adjacency
matrix.  Let $\mD$ be the diagonal matrix of degrees:
\[ D_{i,j} = \begin{cases} \sum_{v} A_{i,v} & i = j \\ 0 & \text{otherwise}.
             \end{cases} \]
The random walk transition matrix is given by $\mP = \mD^{-1} \mA$.
Let $H_{i,j}$ be the hitting time from node $i$ and $j$.
Based on the Markovian nature of a random walk, $H_{i,j}$
must satisfy:
\[ H_{i,j} = 1 + \sum_{v} H_{i,v} P_{v, j}
\qquad \text{ and } \qquad H_{i,i} = 0. \]
That is, the hitting time between
$i$ and $j$ is 1 more than the hitting time between
$i$ and $v$, weighted by the probability of transitioning
between $v$ and $j$, for all $v$.  The minimum non-negative solution $\mH$
that obeys this equation is thus the matrix of hitting times.
The commute time between node $i$ and $j$ is then:
\[ C_{i,j} = H_{i,j} + H_{j,i}. \]
As a matrix $\mC = \mH + \mH^T$, and we refer to $\mC$
as the \emph{commute time} matrix.
An equivalent expression follows from
exploiting a few relationships with the combinatorial graph
Laplacian matrix:
$\mL = \mD - \mA$~\cite{fouss+07}.
Each element
$C_{i,j} = \text{Vol}(G) (L^\dagger_{i,i} - 2 L^\dagger_{i,j} + L^\dagger_{j,j})$
where $\text{Vol}(G)$ is the sum of elements in $\mA$ and
$\mL^\dagger$ is the pseudo-inverse of $\mL$.  The null-space
of the combinatorial graph Laplacian has a well known expression
in terms of the connected components of the graph $G$.
This relationship allows us to write
\[ \mL^\dagger = (\underbrace{\mL + \frac{1}{n} \ve \ve^T}_{\mLhat})^{-1} - \frac{1}{n} \ve \ve^T \]
for \emph{connected} graphs~\cite{saerens},
where $\ve$ is the vector of all
ones, and $n$ is the number of nodes in the graph.
The commute time between nodes in different connected
components is infinite, and thus we only need to consider
connected graphs.
We summarize the notation thus far, and a few subsequent
definitions in Table~\ref{tab:notation}.

\begin{table}[t]
\caption{Notation}
\label{tab:notation}
\begin{small}
\begin{tabularx}{\linewidth}{@{\quad}l@{\quad}X}
\toprule
$\mA$ & the symmetric adjacency matrix for a connected, undirected graph \\
$\mD$ & the diagonal matrix of node degrees\\
$n$ & the number of vertices in $\mA$\\
$\ve$ & the vector of all ones \\
$\ve_i$ & a vector of zeros with a $1$ in the $i$th position\\
$\mL$ & the combinatorial Laplacian matrix of a graph, $\mL = \mD - \mA$\\
$\mLhat$ & the adjusted combinatorial Laplacian, $\mLhat = \mL + \frac{1}{n} \ve \ve^T$\\
$\alpha$ & the damping parameter in Katz\\
$\mK$ & the Katz matrix, $\mK = (\eye - \alpha \mA)^{-1}$ \\				
$\mC$ & the commute time matrix\\
$\mZ$ & a ``general'' matrix, usually $\eye - \alpha \mA$ or $\mLhat$ \\
%$\mP$ & the row-stochastic transition matrix
        %for a random walk on a graph, $\mP = \mD^{-1} \mA$ \\
\bottomrule
\end{tabularx}\end{small}
\vspace{-\baselineskip}
\end{table}

Computing either Katz scores or commute times between all pairs of nodes
involves inverting a matrix:
\[ (\eye - \alpha \mA)^{-1} \quad \text{ or } \quad  (\mL + \frac{1}{n} \ve \ve^T)^{-1}. \]
Standard algorithms for a matrix inverse require
$O(n^3)$ time and $O(n^2)$ memory.  Both of these requirements
are inappropriate for a
large network (see Section~\ref{sec:related} for a brief survey of
existing alternatives).
%Inspired by applications in
%anomalous link detection and recommendation,
%we focus on computing only a single Katz score or commute
%time and on approximating a column of these matrices, where
%the goal is to identify the most related nodes.
Inspired by applications in anomalous link detection
and recommendation, we focus on computing only a single Katz score or
commute time
and on approximating a column of these matrices.
In the former case, our
goal is to find the score for a given pair of nodes and in the latter, it is
to identify the
most related nodes for a given node.
In our vision, the pair-wise algorithms should help in cases
where random pair-wise data is queried, for instance when checking
random network connections, or evaluating user similarity scores
as a user explores a website.  For the column-wise
algorithms, recommending the most similar nodes to a query node
or predicting
the most likely links to a given query node are both
obvious applications.

One way to compute a single score -- what we term the
\emph{pair-wise problem} -- is to find
the value of a bilinear form:
\[ \vu^T \mZ^{-1} \vv , \]
where $\mZ = (\eye - \alpha \mA)$ or $\mZ = \mLhat$.
An interesting approach to estimate these bilinear forms,
and to derive computable upper and lower bounds on the
value, arises from the relationship between the Lanczos/Stieltjes
procedure and a quadrature rule~\cite{Golub93}.  This relationship
and the resulting algorithm for a quadratic form ($\vu^T \mZ^{-1} \vu$)
is described in Section~\ref{sec:mmq-bilinear}.
Prior to that, and because it will form the basis of a few algorithms that we use,
Section~\ref{sec:lanczos} first reviews the properties of the Lanczos method.
We state the pairwise procedure for commute times and Katz
scores in Section~\ref{sec:pairwise-commute}~and~\ref{sec:pairwise-katz}.

The \emph{column-wise problem} is to compute, or approximate, a column
of the matrix $\mC$ or $\mK$.  A column of the commute time matrix is:
\[ \vc_i = \mC \ve_i = \vol(G) [(\ve_i - \ve_v)^T \mLhat^{-1} (\ve_i - \ve_v) : 1 \le v \le n]. \]
A difficulty with this computation is that it requires \emph{all} of the diagonal
elements of $\mLhat^{-1}$, as well as the solution of the
linear system $\mLhat^{-1} \ve_i$.  We can use a property of the Lanczos
procedure and its relationship with the conjugate gradient algorithm to
solve $\mLhat^{-1} \ve_i$ \emph{and} estimate all of the
diagonals of the inverse simultaneously~\cite{Paige-1975-indefinite,Chantas-2008-cgdiag}.

A column of the Katz matrix is $\mK \ve_i$, which corresponds to solving
a single linear system:
\[ \vk_i = \mK \ve_i = (\eye - \alpha \mA)^{-1} \ve_i - \ve_i. \]
Empirically, we observe that the solutions of the Katz linear system
are often localized.  That is, there are only a few large elements in the
solution vector, and many negligible elements.
See Table~\ref{tab:katz-localization} for an example of this
localization over a few graphs.  In order to capitalize
on this phenomenon, we use a generalization of
``push''-style algorithms for personalized PageRank
computing~\cite{mcsherry2005-uniform,%
andersen2006-local,berkhin2007-bookmark}.  These methods
only access the adjacency information for a limited number
of vertices in the graph.  In Section~\ref{sec:columnwise-katz},
we explain the generalization of these methods, the adaptation
to Katz scores, and utilize the theory of coordinate descent
optimization algorithms to establish convergence.
As we argue in that section, these
techniques might also be called ``Gauss-Southwell'' methods,
based on historical precedents.

One of the advantages of Lanczos-based algorithms is that
the convergence is often much faster than a worst-case analysis
would suggest.  This means studying
their convergence by empirical means and on real data sets is important.
 We do so for 17 real-world
networks in Section~\ref{sec:experiments}, ranging in size from
approximation 1,000 vertices to 1,000,000 vertices.  These experiments
highlight both the strengths and weaknesses of our approaches, and
should provide a balanced picture of our algorithms.  In particular,
our algorithms run in seconds or milliseconds -- significantly faster than many 
techniques that use preprocessing
to estimate all of the scores simultaneously, which
can take minutes.

Straightforward approaches based on the conjugate gradient technique
are often competitive with our techniques. However, our algorithms have
other desirable properties, such as upper and lower bounds on the solution
or exploiting sparsity \emph{in the solution vector}, which conjugate gradient does not.
These experiments also shed light on a recent result from \citet{vonLuxburg-2010-commute}
on the relationship between commute time and the degree distribution.

Literature directly related to the problems we study and the techniques we propose is
discussed throughout the paper, in context. However, we have isolated a small set of
core related papers and discuss them in the next section.

In the spirit of reproducible research, we make our data, computational codes, and figure plotting codes available for others:
\url{http://cs.purdue.edu/homes/dgleich/publications/2011/codes/fast-katz/}.

%% file: sec-related.tex
\section{Related work} \label{sec:related}
This paper is about algorithms for computing commute times
and Katz scores over networks with hundreds of thousands
to millions of nodes. Most existing techniques
determine the scores among all pairs of nodes simultaneously
\cite{Acar2009-link-prediction,Wang2007-link-prediction,sarkar-moore07}
(discussed below).
These methods tend to
involve some preprocessing of the graph using a one-time, rather
expensive, computation. We instead focus on quick
estimates of these measures between a single pair of nodes and
between a single node and all other nodes in the graph.
In this vein, a recent paper
\cite{li-etal-sdm10} studies efficient computation of SimRank
\cite{jw02} for a given pair of nodes.

A highly related paper is \citet{bb10}, where
they investigate entries in functions of the adjacency matrix,
such as the exponential,
using quadrature-based bounds.
A priori upper and lower bounds
 are obtained by employing a few Lanczos steps and
the bounds are effectively used to observe
the exponential decay behavior of the exponential of
an adjacency matrix.

In \citet{sarkar-moore07},  an interesting and
efficient approach is proposed for finding approximate nearest neighbors with
respect to a truncated version of the commute time measure.
\citet{spiel08} develop a technique for computing
the effective resistance of all edges (which is proportional to
commute time) in $O(m \log n)$ time.  Both of these procedures
involve some preprocessing.

Standard techniques to approximate Katz scores include truncating
the series expansion to paths of length less than
$\ell_{\text{max}}$~\cite{foster+01,Wang2007-link-prediction} and low-rank
approximation~\cite{cikm03,Acar2009-link-prediction}.  Only the
former technique, when specialized to compute only a pair or
top-$k$ set, has performance comparable to our algorithms.
However, when we tested an adapted algorithm based
on the Neumann series expansion, it required much more work
than the techniques we propose.

As mentioned in the introduction, both commute times and Katz
scores  were studied by \citet{cikm03} for the task of link prediction,
and were found to be effective.
%Since then, work on the link prediction problem has tried to
%exploit metadata about the nodes of a graph, see
%\todo{a few sentences about metadata based link prediction,
%I think of Leskovec's papers at WSDM on supervised random
%walks here, also stuff from Lisa Getoor; or just remove this
%this }
Beyond link prediction, \citet{yen+07} use a commute time kernel based
approach to detect clusters and show that this method
outperforms other kernel based clustering algorithms. The authors
use commute time to define a distance measure between nodes, which
in turn is used for defining a so-called intra-cluster inertia.
Intuitively, this inertia measures how close nodes within a cluster
are to each other.
The algorithm we propose for computing the Katz and commute time score for a
given pair of nodes $x, y$ extends to the case where one wants to
find the aggregate score between a node $x$ and a set of nodes $S$.
Consequently, this work has applications for finding the distance between
a point and a cluster as well as for finding intra-cluster inertia.
For applications to recommender systems, \citet{sarkar08} used their truncated
commute time measure for link prediction over a collaboration
graph and showed that it outperforms personalized
PageRank~\cite{page1999-pagerank}.

%Collaborative filtering is a popular approach to
%recommender systems~\cite{recsys-survey05}. As mentioned
%in the introduction, Sarkar and Moore~\cite{sarkar-moore07}
%motivated commute time in the context of collaborative filtering and proposed an
%interesting and efficient approach for finding approximate nearest
%neighbors with respect to a truncated version of the commute time
%measure.  \citet{saerens} develop an
%application for spectral clustering based on principal component
%analysis of graphs, based on the Euclidean commute time distance
%between nodes, defined as the square root of the average commute
%time.
%They show that the PCA obtained using ECTD has interesting connections
%with spectral graph theory.

%% file: sec-lanczos.tex
\section{The Lanczos Process} \label{sec:lanczos}

The Lanczos algorithm~\cite{lanczos1950-iteration}
is a procedure applied to a symmetric matrix, which works
particularly well when the given matrix is large and sparse.
A sequence of Lanczos iterations can be thought
of as ``truncated'' orthogonal similarity transformations.
Given an $n\times n$ matrix $\mZ$, 
we construct a matrix $\mQ$ with orthonormal columns, one at a time, 
and perform only a small number of steps, say $k$, where $k \ll n$.
The input for the algorithm is the  matrix $\mZ$, an initial
vector $\vq$ and a number of steps $k$. Upon exit, we have an $n
\times (k+1)$ matrix $\mQ_{k+1}$ with orthonormal columns and a $(k+1)
\times k$ tridiagonal matrix $\mT_{k+1,k}$, that satisfy the relationship
\[
\mZ \mQ_k = \mQ_{k+1} \mT_{k+1,k} ,
\]
where $\mQ_k$ is the $n \times k$
matrix that contains the first $k$ columns of $\mQ_{k+1}$, and $ \mT_{k+1,k} =
{\rm tri}(\beta_i,\alpha_i,\beta_i)$ is tridiagonal.

What makes the Lanczos procedure attractive are the good approximation
properties that it has for $k \ll n$. The matrix $\mT_{k+1,k}$ is small when
$k \ll n$, but the eigenvalues of its $k \times k$ upper
part -- a matrix we will refer to as $\mT_k$ in the subsequent section -- approximate the extremal eigenvalues of the large $n \times n$ matrix
$\mZ$.
This can be exploited not
only for eigenvalue computations but also for solving a linear system
\cite{lanczos1952-linear-system,Paige-1975-indefinite}.
Another attractive feature is that the matrix $\mZ$ does not
necessarily have to be provided explicitly;  the algorithm
only uses $\mZ$ via matrix-vector products.

The Lanczos procedure is given in Algorithm~\ref{alg:lanczos}. For expositional purposes we define the core of the algorithm as
Algorithm~\ref{alg:lanczos-step}. We will later incorporate that part into other algorithms -- see Section~\ref{sec:pairwise}.

\begin{figure}
\begin{minipage}{0.45\linewidth}
\begin{algorithm}[H]
\caption{Lanczos($\mZ,\vq,k)$.} \label{alg:lanczos}
\begin{algorithmic} [1]
\STATE $\vq_1 = \vq / \|\vq\|_2, \beta_0 =0, \vq_0 = 0 $
 \FOR{$j=1 \text{ to } k$}
\STATE$\vz= \mZ \vq_j$ \\
\STATE$\alpha_j=\vq_j^T \vz$ \\
\STATE$\vz=\vz-\alpha_j \vq_j - \beta_{j-1} \vq_{j-1}$ \\
\STATE$\beta_j=\|\vz\|_2$\\
\STATE \textbf{if} $\beta_j=0, \vq_{j+1} = 0$ \textbf{and quit} \\
\STATE \textbf{else} $\vq_{j+1}=\vz/\beta_j$ \\
\ENDFOR
\end{algorithmic}
\end{algorithm}
\end{minipage}
\hfil
\begin{minipage}{0.54\linewidth}
\begin{algorithm}[H]
\caption{LanczosStep($\mZ,\vq\itn{-},\vq,\beta\itn{-})$.} \label{alg:lanczos-step}
\begin{algorithmic} [1]
\STATE$\vz= \mZ \vq$ \\
\STATE$\alpha=\vq^T \vz$ \\
\STATE$\vz=\vz-\alpha \vq - \beta\itn{-} \vq\itn{-}$ \\
\STATE$\beta=\|\vz\|_2$\\
\STATE \textbf{if} $\beta=0, \vq\itn{+} = 0$  \\
\STATE \textbf{else} $\vq\itn{+}=\vz/\beta$ \\
\STATE \textbf{return} $(\vq,\alpha,\beta)$
\end{algorithmic}
\end{algorithm}
\end{minipage}
\end{figure}

%% file: sec-pairwise.tex
\section{Pairwise Algorithms} \label{sec:pairwise}

Consider the commute time and Katz score between a single pair
of nodes:
\[
    C_{i,j} = \text{Vol}(G)(\ve_i - \ve_j)^T \mL^\dagger (\ve_i - \ve_j)
		\quad
		\text{ and }
		\quad
		K_{i,j} = \ve_i^T (\eye - \alpha \mA)^{-1} \ve_j - \delta_{i,j}.
\]
	In these expressions, $\ve_i$ and $\ve_j$ are vectors of zeros
with a $1$ in the $i$th and $j$th position, respectively; and
$\delta_{i,j}$ is the Kronecker delta function.
A straightforward means of computing them is to solve
the linear systems
\[ \mLhat^{-1} \vy = \ve_i - \ve_j
    \quad
		\text{ and }
		\quad
		(\mI - \alpha \mA) \vx = \ve_j.
\]
Then
$C_{i,j} = \text{Vol}(G)(\ve_i - \ve_j)^T \vy$
and $K_{i,j} = \ve_i^T \vx - \delta_{i,j}$.
It is possible to compute the pair-wise scores by solving these linear systems. In what follows, we show
how a technique combining the Lanczos iteration and a quadrature
rule~\cite{Golub93,Golub97}
produces the pair-wise commute time score or
the pair-wise Katz score
\emph{as well as} upper and lower bounds on the estimate.

\subsection{Matrices moments and quadrature}
\label{sec:mmq} \label{sec:mmq-bilinear}

Both of the pairwise computations above are instances of the general problem of estimating a bilinear
form:
\[ \vu^T f(\mZ) \vv , \]
where $\mZ$ is symmetric positive definite (for Katz, this
occurs by restricting the value of $\alpha$, and for commute times,
the adjusted Laplacian $\mLhat$ is always positive definite),
and $f(x)$ is an analytic function on the region containing
the eigenvalues of $\mZ$.  The only function $f(x)$ we use
in this paper is $f(x) = \frac{1}{x}$, although we treat the
problem more generally for part of this section.

\citet{Golub93, Golub97}
introduced elegant computational techniques for evaluating such bilinear
forms. They provided a solid mathematical framework and a rich
collection of possible applications.
These techniques are well known in the numerical linear algebra
community, but they do not seem to have been used in data mining
problems.
We adapt this methodology to the pairwise score problem, and
explain how to do so in an efficient manner in a large scale
setting. The algorithm has two main components, combined together:
Gauss-type quadrature rules for evaluating definite integrals, and
the Lanczos algorithm for partial reduction to symmetric tridiagonal
form.  In the following discussion, we treat the case of $\vu = \vv$.
This form suffices thanks to the identity
\[
\vu^T f(\mZ) \vv =
\frac{1}{4} \left[
(\vu + \vv)^T f(\mZ) (\vu + \vv) - (\vu - \vv)^T f(\mZ) (\vu - \vv)
\right].
\]

Because $\mZ$ is symmetric positive definite,
it has a unitary spectral decomposition,
$ \mZ = \mQ \mLambda \mQ^T, $ where $\mQ$ is an
orthogonal matrix whose columns are eigenvectors of $\mQ$ with unit
$2$-norms, and $\mLambda$ is a diagonal matrix with the eigenvalues of
$\mQ$ along its diagonal.
We use this decomposition only
for the derivation that follows,
it is never computed in our algorithm.
Given this decomposition, for
any analytic function $f$,
\[
 \vu^T f(\mZ) \vu = \vu^T \mQ f(\mLambda) \mQ^T \vu =
\sum_{i=1}^n  f(\lambda_i) \tilde{u}_i^2, \]
where
$\vec{\tilde{u}}=\mQ^T \vu$.
The last sum is equivalent to the Stieltjes integral
\begin{equation}
\label{eq:int}
\vu^T f(\mZ) \vu =
\int_{\lmin}^{\lmax} f(\lambda) \, d\omega(\lambda).
\end{equation} Here
$\omega(\lambda)$ is a piecewise
constant measure, which is monotonically increasing,
and its values depend directly on the eigenvalues of $\mZ$.
Both $\lmin$ and $\lmax$ are values that are 
lower and higher than the extremal eigenvalues of $\mZ$, respectively.
Let
\[
0 < \lambda_1 \le \lambda_2 \le \ldots \le \lambda_n
\]
be the eigenvalues of $\mZ$.  Note that $\lmin < \lambda_1$
and $\lmax > \lambda_n$.
Now, $\omega(\lambda)$ takes the following form:
\[
\omega(\lambda) = \begin{cases}
0 & \lambda < \lambda_1 \\
\sum_{j=1}^i \tilde{u}_j^2 & \lambda_i \le \lambda < \lambda_{i+1} \\
\sum_{j=1}^n \tilde{u}_j^2 & \lambda_n \le \lambda.
\end{cases} \]

The first of Golub and Meurant's key insights is that we
can compute an approximation for an integral of the form
\eqref{eq:int} using a quadrature rule, e.g.,
\[ \int_{\lmin}^{\lmax} f(\lambda) \, d\omega(\lambda)
\approx \sum_{j=1}^N f(\eta_j) \omega_j
\]
where $\eta_j, \omega_j$ are the nodes and weights
of a Gaussian quadrature rule.
The second insight is that the
Lanczos procedure \emph{constructs
the quadrature rule itself}.
Since we use a quadrature rule, an estimate of the error is readily
 available, see for example~\citet{davis1984-integration}.
More importantly, we can use variants of the Gaussian
quadrature to obtain
both lower and upper bounds and ``trap'' the value of the element of
the inverse that we seek between these bounds.

The ability to
estimate bounds for the value is powerful and provides
effective stopping criteria for the algorithm -- we shall see
this in the experiments in Section~\ref{sec:experiments-pairwise-commute}.
It is important to
note that such component-wise bounds cannot be easily obtained if we were to
extract the value of the element from a column of the inverse, by
solving the corresponding linear system, for example.
Indeed, typically for the solution of a linear system, norm-wise bounds
are available, but obtaining bounds pertaining to the components
of the solution is significantly more challenging and results of
this sort are harder to establish.
It should also be noted that bounds of the sort discussed here cannot
be obtained for general non-symmetric matrices.

Returning to the procedure, let $f(\lambda)$ be a function where the $(2k+1)$st derivative has a negative
sign for all $\lmin < \lambda < \lmax$.  Note that $f(\lambda) = 1/\lambda$ satisfies
this condition because all odd derivatives are negative when $\lambda > 0$.
As a high level algorithm, the Golub-Meurant procedure for
estimating bounds
\[ \underline{b} \le \vu^T f(\mZ) \vu \le \overline{b} \]
is given by the following steps.
\begin{compactenum}
\item Let $\sigma = \normof{\vu}$.
\item Compute $\mT_k$ from $k$ steps of the Lanczos procedure
applied to $\mZ$ and $\vu/\sigma$.
\item Compute $\mat{\underline{T}}_k$, which the matrix $\mT_k$ extended
with another row and column crafted to add the eigenvalue $\lmax$ to
the eigenvalues of $\mT_k$.  This new matrix is still tridiagonal.
\item Set $\underline{b} = \sigma^2 \ve_1^T f(\mat{\underline{T}}_k) \ve_1$.
This estimate corresponds to a $(k+1)$-point Gauss-Radau rule
with a prescribed point of $\lmax$.
\item Compute $\mat{\overline{T}}_k$, which the matrix $\mT_k$ extended
with another row and column crafted to add the eigenvalue $\lmin$ to
the eigenvalues of $\mT_k$.  Again, this new matrix is still tridiagonal.
\item Set $\overline{b} = \sigma^2 \ve_1^T f(\mat{\overline{T}}_k) \ve_1$.
This estimate corresponds to a $(k+1)$-point Gauss-Radau rule
with a prescribed point of $\lmin$.
\end{compactenum}
Based on the theory of Gauss-Radau quadrature, the fact that these
are lower and upper bounds on the quadratic form $\vu^T f(\mZ) \vu$
follows because the \emph{sign} of the error term changes when
prescribing a node in this fashion.  See \citet[Theorem~6.4]{Golub-2010-mmq}
for more information.
As $k$ increases, the upper and lower bounds converge.

While this form of the algorithm is convenient for understanding
the high level properties and structure of the procedure, it is
not computationally efficient.  If $f(\lambda) = 1/\lambda$ and
if we want to compute a more accurate estimate by increasing $k$, then
we need to solve two inverse eigenvalue problems (steps 3 and 5), and
solve two linear systems (steps 4 and 6).  Each of these steps involves
$O(k)$ work because the matrices involved are tridiagonal.  However,
a \emph{constant-time} update procedure is possible.  The set of
operations to efficiently update $\blower$ and $\bupper$ after
a Lanczos step (Algorithm~\ref{alg:lanczos-step}) is given
by Algorithm~\ref{alg:mmqstep}.  Please see
\citet{Golub97} for an explanation of this procedure.
Using Algorithms~\ref{alg:lanczos-step}~and~\ref{alg:mmqstep}
as sub-routines, it is now straightforward to state the
pairwise commute time and Katz procedures.

\begin{algorithm}[p]
\caption{MMQStep~\cite[Algorithm GQL]{Golub97}} \label{alg:mmqstep}
\begin{algorithmic} [1]
\REQUIRE $\alpha, \beta\lasti, \beta,
b\lasti, c\lasti, d\lasti,
\dlower\lasti, \dupper\lasti$

\STATE
  $b = b\lasti +
	  \frac{\beta\lasti^{2} c\lasti^{2}}
		     {d\lasti(\alpha d\lasti-\beta\lasti^{2})}$;
		\quad
		$c = c\lasti \frac{\beta\lasti}{d\lasti}$;
		\quad 
    $d = \alpha - \frac{\beta\lasti^{2}}{d\lasti}$

 \STATE
   $\dupper = \alpha - \lmin - \frac{\beta\lasti^{2}}{\dupper\lasti}$; \quad
   $\dlower = \alpha - \lmax - \frac{\beta\lasti^{2}}{\dlower\lasti}$

 \STATE
   $\omegaupper = \lmin + \frac{\beta^{2}}{\dupper}$; \quad
   $\omegalower = \lmax + \frac{\beta^{2}}{\dlower}$
	
 \STATE
   $\bupper = b + \frac{\beta^2 c^2}{d(\omegaupper d - \beta^2)}$; \quad
   $\blower = b + \frac{\beta^2 c^2}{d(\omegalower d - \beta^2)}$
	
 \ENSURE $(\bupper, \blower)$ and
 ($b, c, d, \dupper, \dlower$)
\end{algorithmic}
\end{algorithm}

\subsection{Pairwise commute scores}
\label{sec:pairwise-commute}

The bilinear form that we need to estimate
a commute time is
\[
b = (\ve_i - \ve_j)^T \mLhat^{-1} (\ve_i -\ve_j).
\]
For this problem, we apply
Algorithm~\ref{alg:lanczos-step} to step
through the Lanczos process and then use
Algorithm~\ref{alg:mmqstep}
to update the upper and lower bounds on the score.
This combination is explicitly described in
Algorithm~\ref{alg:pairwise-commute}.
Note that we do not need to apply the final correction
with $\frac{1}{n} \ve \ve^T$ because  $\ve^T (\ve_i - \ve_j) = 0$.

\begin{algorithm}[p]
\caption{Pairwise Score Bounds for commute time}
\label{alg:commute-pairwise} \label{alg:pairwise-commute}
\label{alg:mmq-commute}
\begin{algorithmic} [1]
\REQUIRE $\mL$ (Laplacian matrix); $i, j$ (pairwise coordinate);
$\lmin, \lmax$ (bounds where $\lmin < \lambda(\mL) < \lmax$);
$\tau$ (stopping tolerance)
\ENSURE $\kappamin, \kappamax$ where
$\kappamin < (\ve_i - \ve_j)^T \mL^{\dagger} (\ve_i - \ve_j) < \kappamax$
\STATE (Initialize Lanczos)
  $\sigma = \sqrt{2}, \vq_{-1} = 0, \vq_0 = (\ve_i - \ve_j)/\sigma, \beta_0 = 0$
\STATE (Initialize MMQStep)
	$b_0=0, c_0=1, d_0=1, \dupper_0=1, \dlower_0=1$
\FOR{$j=1,...$}
 \STATE
    Set $(\vq_j,\alpha_j,\beta_j)$ from
    LanczosStep$(\mLhat,\vq_{j-2},\vq_{j-1},\beta_{j-1})$
		
 \STATE
    Set $(\bupper, \blower)$ and
		$(b_j, c_j, d_j, \dupper_j, \dlower_j)$
		from \\MMQStep$(\alpha_j, \beta_{j-1}, \beta_j,
			b_{j-1}, c_{j-1}, d_{j-1}, \dlower_{j-1}, \dupper_{j-1})$.
			
	\STATE $\kappamin = \sigma^2 \blower$; $\kappamax = \sigma^2 \bupper$
	\STATE \textbf{if} $\kappamax - \kappamin < \tau$, \textbf{stop}
	
 \ENDFOR
\end{algorithmic}
\end{algorithm}

\subsection{Pairwise Katz scores}
\label{sec:pairwise-katz}

The bilinear form that we need to estimate
for a Katz score is
\[ b = \ve_i^T (\eye - \alpha \mA)^{-1} \ve_j. \]
Recall that we use the identity:
\[ b = \frac{1}{4}
\bigl[ \underbrace{(\ve_i + \ve_j)^T (\eye - \alpha \mA)^{-1} (\ve_i + \ve_j)}_{=g}
- \underbrace{(\ve_i - \ve_j)^T (\eye - \alpha \mA)^{-1} (\ve_i - \ve_j)}_{=h} \bigr] \]
In this case, we apply the combination of LanczosStep and MMQStep
to estimate $\glower \le g \le \gupper$ and $\hlower \le h \le \hupper$.
Then $\frac{1}{4}(\glower - \hupper) \le b \le \frac{1}{4} (\gupper - \hlower)$.

\begin{algorithm}[p]
\caption{Pairwise Score Bounds for Katz}
\label{alg:katz-pairwise} \label{alg:pairwise-katz}
\begin{algorithmic} [1]
\REQUIRE $\mA$ (adjacency matrix); $\alpha$ (the Katz damping factor);
$i, j$ (pairwise coordinate); $\lmin, \lmax$ (bounds where $\lmin < \lambda(\eye - \alpha \mA) < \lmax$);
$\tau$ (stopping tolerance)
\ENSURE $\rhomin, \rhomax$ where $\rhomin < (\eye - \alpha \mA)^{-1}_{i,j} < \rhomax$

\STATE (Initialize Lanczos for $g$)
  $\sigma = \sqrt{2}, \vq_{-1} =0, \vq_0 = (\ve_i + \ve_j)/\sigma, \beta_0^g = 0$
\STATE (Initialize Lanczos for $h$)
  $\vu_{-1} = 0, \vu_0 = (\ve_i - \ve_j)/\sigma, \beta_0^h = 0$
\STATE (Initialize MMQStep for $g$)
	$b_0^g=0, c_0^g=1, d_0^g=1, \dupper_0^g=1, \dlower_0^g=1$
\STATE (Initialize MMQStep for $h$)
  $b_0^h=0, c_0^h=1, d_0^h=1, \dupper_0^h=1, \dlower_0^h=1$
\FOR{$j=1,...$}
 \STATE
    Set $(\vq_j,\alpha_j^g,\beta_j^g)$ from
    LanczosStep$((\eye - \alpha \mA),\vq_{j-2},\vq_{j-1},\beta_{j-1}^f)$
 \STATE
    Set $(\vu_j,\alpha_j^h,\beta_j^h)$ from
    LanczosStep$((\eye - \alpha \mA),\vu_{j-2},\vu_{j-1},\beta_{j-1}^h)$		
 \STATE
    Set $(\gupper, \glower)$ and
		$(b_j^g, c_j^g, d_j^g, \dupper_j^g, \dlower_j^g)$
		from \\MMQStep$(\alpha_j^g, \beta_{j-1}^g, \beta_j^g,
			b_{j-1}^g, c_{j-1}^g, d_{j-1}^g, \dlower_{j-1}^g, \dupper_{j-1}^g)$.
	
	\STATE
    Set $(\hupper, \hlower)$ and
		$(b_j^h, c_j^h, d_j^h, \dupper_j^h, \dlower_j^h)$
		from \\ MMQStep$(\alpha_j^h, \beta_{j-1}^h, \beta_j^h,
			b_{j-1}^h, c_{j-1}^h, d_{j-1}^h, \dlower_{j-1}^h, \dupper_{j-1}^h)$.
			
	\STATE $\rhomin = \sigma^2/4 (\glower - \hupper)$;
	  $\rhomax = \sigma^2/4 (\gupper - \hlower)$
	\STATE \textbf{if} $\rhomax - \rhomin < \tau$, \textbf{stop}
	
 \ENDFOR

\end{algorithmic}
\end{algorithm}

%% file: sec-topk.tex
\section{Column-wise algorithms} \label{sec:rowwise} \label{sec:columnwise}

Whereas the last section used a single procedure to derive two algorithms, in this
section, we investigate two different procedures: one for commute time and a different
procedure for Katz scores. The reason behind this difference is that, as mentioned
in the introduction, computing a column
of the commute time matrix cannot be stated as the solution of a single linear
system:
\[ \vc_i = \mC \ve_i = \vol(G) [(\ve_i - \ve_v)^T \mLhat^{-1} (\ve_i - \ve_v) : 1 \le v \le n]. \]
Computing this column requires \emph{all} of the diagonal elements of the inverse.
In contrast, a column of the Katz matrix is just the solution of a linear
system:
\[
\vk_i = \mK \ve_i = (\eye - \alpha \mA)^{-1} \ve_i - \ve_i.
\]
For this computation, we exploit an empirical localization
property of these columns.

\subsection{Column-wise commute times}
\label{sec:columnwise-commute}

A straightforward way to compute an entire column of the commute time matrix would
require solving $n$ separate linear systems: one to get both $\mLhat^{-1} \ve_i$
and $\mLhat^{-1}_{i,i}$, and the other $n-1$ to get $\mLhat^{-1}_{j,j}$
for $i \not= j$.  Neither solving each system independently, nor using
a \emph{multiple right hand side} algorithm
\cite{OLeary-1980-block-cg}, will easily
yield an efficient procedure.  Both of these approaches generate far
too much extraneous information.  In fact, we only need one linear system
solve, and the diagonal elements of the pseudo-inverse.  Thus, any procedure to compute
or estimate $\diag(\mL^{\dagger})$ provides a tractable algorithm.

One such procedure arises, again, from the Lanczos method.
It was originally described by
\citet{Paige-1975-indefinite}, and is explained in more detail in
\citet{Chantas-2008-cgdiag}.
Suppose we want to compute $\diag(\mLhat^{-1})$.
If the Lanczos algorithm runs to
completion in exact arithmetic, then we have:
\[ \mLhat = \mQ \mT \mQ^T \qquad \text{ and } \qquad
   \mLhat^{-1} = \mQ \mT^{-1} \mQ^T.  \]
Let $\mT = \mR \mR^T$ be a Cholesky factorization of $\mT$.  If we substitute
this factorization into the expression for the inverse, then
$\mLhat^{-1} = \mV \mR^{-T} \mR^{-1} \mV^T$.  Now, let $\mW = \mV \mR^{-T}$.
Note that $\mLhat^{-1} = \mW \mW^T$.  As a notational convenience, let $\vw_k$ be the
$k$th column of $\mW$. Consequently,
\[ \diag(\mE^{-1}) = \sum_{k=1}^n \vw_k \circ \vw_k  \]
where $\vw_k \circ \vw_k$ is the Hadamard (element-wise) product: 
$[\vw_k \circ \vw_k]_i = \vw_{k,i}^2$.  If we implement CG based
on the Lanczos algorithm as explained in \citet{Paige-1975-indefinite},
then the vector $\vw_k$ is computed as part of the standard algorithm,
and is available at no additional cost. This idea is implemented
in the \texttt{cgLanczos.m} code~\cite{Saunders-2007-cgLanczos},
which we use in our experiments. Please see
\citet{Chantas-2008-cgdiag} for a detailed account of this derivation
including the diagonal estimate.

Based on advice from the author of the \texttt{cgLanczos} code,
we added local reorthogonalization
to the Lanczos procedure.  This addition requires a few extra vectors of
memory, but ensures greater orthogonality in the computed Lanczos vectors
$\vq_k$.  Also, based on advice from the author, we use the following
preconditioned linear system:
\[ \mD^{-1/2} \mLhat \mD^{-1/2} \vy = \mD^{-1/2} \ve_i. \]
If $\vf$ is the estimate of the diagonals of $(\mD^{-1/2} \mLhat \mD^{-1/2})^{-1}$,
then $\mD^{-1} \vf$ is the estimate of the diagonals of $\mLhat^{-1}$.
Using this preconditioned formulation, the algorithm converged much more quickly
than without preconditioning.  In summary, this
approach to estimate the column-wise commute times $\vc_i$ is:
\begin{compactenum}
\item Solve $\mD^{-1/2} \mLhat \mD^{-1/2} \vy = \mD^{-1/2} \ve_i$
      using \texttt{cgLanczos.m} \\ to get both $\vy$ and
			$\vf \approx \diag\left( (\mD^{-1/2} \mLhat \mD^{-1/2})^{-1} \right)$.
\item Set $\vx = \mD^{-1/2} \vy - \frac{1}{n}\ve \approx \mL^{\dagger} \ve_i$.			
\item Set $\vg = \mD^{-1} \vf - \frac{1}{n}\ve \approx \diag( \mL^{\dagger} )$.
\item Output $\vc_i \approx \vg + x_i\ve - 2 \vx$.
\end{compactenum}
We refrain from stating this as a formal algorithm because the majority
of the work is in the $\texttt{cgLanczos.m}$ routine.

\subsection{Column-wise Katz scores}
\label{sec:columnwise-katz}

In this section, we show how to adapt techniques for rapid
personalized PageRank computation \cite{mcsherry2005-uniform,
andersen2006-local,berkhin2007-bookmark}
to the problem of computing a column of the Katz matrix.
Recall that such a column is given by the solution of a single
linear system:
\[ \vk_i = \mK \ve_i = (\eye - \alpha \mA)^{-1} \ve_i - \ve_i. \]
The algorithms for personalized PageRank
exploit the graph structure by accessing the edges
of individual vertices, instead of accessing the graph via
a matrix-vector product.  They
are ``local'' because they only access the adjacency information
of a small set of vertices and need not explore the majority
of the graph.  Such a property is useful when the solution of
a linear system is localized on a small set of elements.

Localization is a term with a number of interpretations.  Here, we use
it to mean that the vector becomes sparse after rounding small
elements to 0.  A nice way of measuring this property is to look
at the participation ratios~\cite{Farkas-2001-spectra}.
Let $\vk$ be a column of the
Katz matrix, then the participation ratio of $\vk$ is
\[ p = \frac{(\sum_j k_j^2)^2}{\sum_j k_j^4}. \]
This ratio measures the number of effective non-zeros of the
vector.
If $\vk$ is a uniform vector, then $p = n$, the size of the vector.
If $\vk$ has only a single element, then $p = 1$, the number of
states occupied.  For a series of graphs we describe more formally
in Section~\ref{sec:data}, we show the statistics of
some participation ratios in Table~\ref{tab:katz-localization}.
We pick columns of
the matrix in two ways: (i) randomly and (ii) from the degree distribution to
ensure we choose both high, medium, and low degree vertices.
See Section~\ref{sec:runtime} for a more formal
description about how we pick columns; we use the ``hard alpha''
value of Katz described in the experiments section.
The results show that number of effective non-zeros is always
less than 10,000, even when the graph has 1,000,000 vertices.
Usually, it is even smaller.  Our forthcoming algorithms
exploit this property.

\begin{table}
\caption{Participation ratios for Katz scores.  These results
demonstrate that the columns of the Katz matrix are highly localized.
In the worst case, there are only a few thousand large elements
in a vector, compared with the graph size of a few hundred thousand vertices. }
\label{tab:katz-localization}
\small
\begin{tabularx}{\linewidth}{lXXXXXX}
\toprule
Graph & Vertices & Avg.~Deg. & \multicolumn{4}{l}{Participation Ratios} \\
\cmidrule{4-7}
& & & Min & Mean & Median & Max \\
\midrule
               tapir &    1024 &   5.6 & 4.2 &   12.0 &   11.8 &   35.8 \\
     stanford-cs-sym &    2759 &   7.4 & 1.0 &   26.3 &   23.5 &  274.1 \\
             ca-GrQc &    4158 &   6.5 & 1.0 &   27.4 &   34.0 &   84.2 \\
           wiki-Vote &    7066 &  28.5 & 1.2 &  248.8 &  291.6 &  342.6 \\
            ca-HepTh &    8638 &   5.7 & 1.0 &   23.5 &   29.8 &   82.1 \\
            ca-HepPh &   11204 &  21.0 & 1.0 &  160.7 &  256.1 &  268.5 \\
           Stanford3 &   11586 &  98.1 & 1.1 & 1509.5 & 1657.8 & 1706.4 \\
          ca-AstroPh &   17903 &  22.0 & 1.0 &  167.5 &  219.2 &  290.8 \\
          ca-CondMat &   21363 &   8.5 & 1.0 &   71.0 &   85.6 &  204.6 \\
         email-Enron &   33696 &  10.7 & 1.0 &  203.0 &  262.5 &  598.6 \\
       soc-Epinions1 &   75877 &  10.7 & 1.0 &  299.2 &  455.6 &  526.0 \\
    soc-Slashdot0811 &   77360 &  12.1 & 1.0 &  320.4 &  453.3 &  495.8 \\
               arxiv &   86376 &  12.0 & 1.0 &  121.1 &  137.9 &  508.6 \\
                dblp &   93156 &   3.8 & 1.0 &   50.0 &   25.2 &  258.9 \\
         email-EuAll &  224832 &   3.0 & 1.0 &  237.7 &  276.7 & 7743.7 \\
             flickr2 &  513969 &  12.4 & 1.0 &  592.3 & 1104.9 & 1414.9 \\
      hollywood-2009 & 1069126 & 105.3 & 2.0 & 1696.0 & 2433.8 & 3796.0 \\
\bottomrule
\end{tabularx}
\end{table}			

The basis of these personalized PageRank
algorithms is a variant on the Richardson
stationary method for solving a linear system~\cite{v62}.
Given a linear system $\mZ \vx = \vb$, the Richardson iteration is
\[  \vx\itn{k+1} = \vx\itn{k} + \ \vr\itn{k}, \]
where $\vr\itn{k} = \vb - \mZ \vx\itn{k}$
is the residual vector at the $k$th iteration. 
While updating $\vx\itn{k+1}$ is a linear time operation,  computing
the next residual requires another matrix-vector product.
To take advantage of the graph structure,
the personalized PageRank algorithms \cite{mcsherry2005-uniform,andersen2006-local,berkhin2007-bookmark}
propose the following change:
 do not update $\vx\itn{k+1}$ with the entire residual, and instead change
 only a single component of $\vx$.
Formally,
 $\vx\itn{k+1} = \vx\itn{k} + r\itn{k}_j \ve_j$,
 where $\ve_j$ is a vector of all zeros, except for a single $1$ in the $j$th position,
and $r\itn{k}_j$ is the $j$th component of the residual vector.
 Now, computing the next residual involves accessing a single column of the matrix
 $\mZ$:
 \[ \vr\itn{k+1} = \vb - \mZ \vx\itn{k+1} = \vb - \mZ(\vx\itn{k} +  r\itn{k}_j \ve_j)
    = \vr\itn{k} + r\itn{k}_j \mZ \ve_j.
 \]
 Suppose that $\vr,$ $\vx,$ and $\mZ \ve_j$ are sparse, then this update introduces
 only a small number of new nonzeros into both $\vx$ and the new residual $\vr$.
 If $\mZ = (\eye - \alpha \mA)$, as in the case of Katz, then each
 column is sparse, and thus
 keeping the solution and residual sparse is a natural choice for graph algorithms
 where the solution $\vx$ is localized (i.e., many components of $\vx$ can be rounded to 0
 without dramatically changing the solution).
 By choosing the element $j$ based on the largest entry in the sparse residual
 vector (maintained in a heap),
 this algorithm often finds a good approximation to the largest entries
 of the solution vector $\vx$ while exploring only a small subset of the graph.
 The resulting procedure is presented in Algorithm~\ref{alg:katz-columnwise}.
 For reasons that will become clear below, we call this procedure the
 Gauss-Southwell algorithm.  When experimenting with this
 method, we found that picking
 elements from the heap proportional to $\mD^{-1} \vr$ instead of $\vr$
 yielded convergence with fewer total edges explored,
 mirroring
 the results in \citet{andersen2006-local}.  We use this version in
 all of our experiments, although we state all the formal convergence
 results for the simple choice of residual $\vr$.

\begin{algorithm}
\caption{Column-wise Katz scores (via the Gauss-Southwell algorithm)} \label{alg:katz-columnwise}
\begin{algorithmic}[1]
\REQUIRE $\mA$ (the adjacency matrix), $\alpha$ (the Katz damping factor),
$i$ (the desired column), $\tau$ (a stopping tolerance).
\ENSURE $\vx$ (an approximate solution of $(\eye - \alpha \mA)^{-1} \ve_i )$
\STATE Set $\vx = 0$, $\vr = 0$
\STATE Let $\set{H}$ be a heap over the non-zero entries of $\vr$ larger than $\tau$.
\STATE Set $r_i = 1$, update $\set{H}$
\WHILE{$\set{H}$ is not empty}
  \STATE Set $j$ as the index of the largest element in $\set{H}$
	\STATE \textbf{if} $r_j < \tau$ \textbf{then} quit.
	\STATE $\eta = r_j$
	\STATE $x_j \leftarrow x_j + \eta$
	\STATE $r_j \leftarrow 0$, remove $j$ from $\set{H}$
	\FOR{$u$ where $A_{j,u} > 0$}
	  \STATE $r_u \leftarrow r_u + \alpha \eta$
		\STATE \textbf{if} $r_u > \tau$ \textbf{then} insert $j$ in $\set{H}$ or update $\set{H}$.
	\ENDFOR
\ENDWHILE	
\STATE $x_i \leftarrow x_i - 1$
\end{algorithmic}
\end{algorithm}

 Let $\dmax$ be the maximum degree of a node in the graph, then each
 iteration takes $O(\dmax \log n)$ time.
We analyze the convergence of this algorithm for Katz scores in two
stages.  In the first case, when $\alpha < 1/\dmax$, then the convergence
theory of this method for personalized PageRank \emph{also} shows
that it converges for Katz scores.  This fortunate occurrence results
from the equivalence of Katz scores and the general formulation of
PageRank adopted by \citet{mcsherry2005-uniform} in this setting.
In the second case, when $\alpha < 1/\sigma_{\max}(\mA)$, then
$(\eye - \alpha \mA)$ is still symmetric positive definite, and
the Richardson algorithm converges.  To show convergence in this case,
we will utilize an equivalence between this algorithm and a coordinate
descent method.

For completeness, we show a precise convergence result when $\alpha < 1/\dmax$.
The key observation here is that the residual $\vr$ is always non-negative and
that the sum of the residual ($\ve^T \vr$) is monotonically decreasing.
To show convergence, we need to bound this sum by a function that converges
to 0.

Consider the
algorithm applied to $(\meye - \alpha \mA) \vx = \ve_i$.
From step $k$ to step $k+1$,
the algorithm sets
\[ \vx\itn{k+1} = \vx\itn{k} + \eta \ve_j; \qquad
   \vr\itn{k+1} = \vr\itn{k} + \eta (\eye - \alpha \mA) \ve_j. \]
First note that $\alpha < 1/\dmax$ implies $r\itn{k+1}_i \ge 0$ given $r\itn{k}_
i \ge 0$.  This bound now implies that
$\vx\itn{k+1}_i \ge 0$ when $\vx\itn{k}_i \ge 0$.
Since these conditions hold for the initial conditions, $\vx\itn{0} = 0$ and $\vr\itn
{0} = \ve_q$, they remain true throughout the iteration.  Consequently, we can use the
\emph{sum} of $\vr\itn{k}$ as the 1-norm of this vector, that is, $\ve^T \vr\itn{k+1} =
\normof[1]{\vr\itn{k+1}}$.
It is now straightforward to analyze the convergence of
this sum:
\[ \ve^T \vr\itn{k+1} = \ve^T \vr\itn{k} - \eta
                         + \alpha \eta \ve^T \mA \ve_i. \]
At this point, we need the bound that $\eta = r\itn{k}_j \ge (1/n) \ve^T \vr\itn{k}$, which
follows immediately from the fact that $r\itn{k}_j$ is the
largest element in $\vr\itn{k}$.
Also, $\ve^T \mA \ve_i \le d_{\max}$.  Thus, we conclude:
\begin{remark} \label{rem:katz-simple}
If $\alpha < 1/\dmax$, then the 1-norm of the residual in the Gauss-Southwell
iteration applied to the Katz linear system satisfies
\[
\normof[1]{\vr\itn{k+1}}
  \le \biggr(1 - \frac{ 1- \alpha \dmax}{n}\biggl) \normof[1]{\vr\itn{k}}
	\le \biggr(1 - \frac{ 1- \alpha \dmax}{n}\biggl)^{k} .
\]
\end{remark}

In the second case, when $1/\dmax < \alpha < 1/\sigma_{\max}(\mA)$,
then the Gauss-Southwell iteration in Algorithm~\ref{alg:katz-columnwise}
still converges, however, the result is more intricate than
the previous case because the sum of the residual does not converge
monotonically. As we shall see, treating this linear system as an optimization
problem provides a way to handle this case.  Let $\mZ$ be symmetric
positive definite.  We first show that the Gauss-Southwell algorithm
is a coordinate descent method for the convex problem
\[ \MIN{}{\frac{1}{2} \vx^T \mZ \vx - \vx^T \vb = f(\vx).} \]
The gradient of this problem is $\mZ \vx - \vb$, hence
a stationary point is the solution of the linear system,
and the global minimizer.  In this framework, the
Richardson method is a gradient descent method.
Let $\vg\itn{k}$ be the gradient at step $k$,
$\vg\itn{k} = \mZ \vx\itn{k} - \vb$ then
\[ \vx\itn{k+1} = \vx\itn{k} - \vg \]
is exactly the Richardson step.

Now consider a standard
coordinate descent method.  Such methods usually minimize
the function in the $j$th coordinate \emph{exactly}.
Formally, they find
\[ \vx\itn{k+1} = \vx\itn{k} + \gamma\itn{k} \ve_j \]
where
\[ \gamma\itn{k} = \argmin_\gamma f(\vx\itn{k} + \gamma \ve_j). \]
Solving this system produces the choice
\[ \gamma\itn{k} = \frac{b_j - (\mZ \ve_j)^T \vx\itn{k} }{Z_{j,j}}. \]
Note that in terms of the optimization problem the Gauss-Southwell
algorithm generates
\[ \gamma\itn{k}_S = r_j\itn{k} = (b_j - \vz_j^T \vx\itn{k}). \]
The two methods are equivalent if the diagonals of $\mA$ are 1.
Consequently, we have:
\begin{lemma}
The Gauss-Southwell method for $\mZ \vx = \vb$
with $Z_{i,i} = 1$
is equivalent to a coordinate gradient descent method
for the function $f(\vx) = (1/2) \vx^T \mZ \vx - \vx^T \vb$.
\end{lemma}

To produce a convergent algorithm, we must now specify
how to choose the descent direction $j$.
\begin{theorem} \label{thm:katz-coordinate}
Let $\mZ$ be symmetric positive definite with $Z_{i,i} = 1$.
Then the Gauss-Southwell method for $\mZ \vx = \vb$ and
$j\itn{k} = \argmax_i \absof{r_i\itn{k}}$ or with $j\itn{k}$
chosen cyclically ($j\itn{1} = 1, j\itn{k+1} = j\itn{k}+1 \mod n$)
is convergent.
\end{theorem}
\begin{proof}
This result follows from the convergence of the coordinate
descent method~\cite[Theorem 2.1]{Luo1992-coordinate-descent}
with these two update rules.  The first is also known
as the Gauss-Southwell rule.
\end{proof}

This proof demonstrates that, as long as $A_{i,i} = 0$ for all
the diagonal entries of the adjacency matrix, then
Algorithm~\ref{alg:katz-columnwise} will converge when
$(\eye - \alpha \mA)$ is positive definite, that is, when
$\alpha < 1/\sigma_{\max}(\mA)$.  We term this algorithm
a Gauss-Southwell procedure because the choice of $j$
in the algorithm is given by the Gauss-Southwell rule.

%% file: sec-experiments.tex
\section{Experimental Results} \label{sec:experiments}

The previous sections showed three algorithms based on the
Lanczos method, and showed the theoretical convergence
of the column-wise Katz algorithm.  In this section,
we investigate these algorithms numerically.  Algorithms
based on the Lanczos method, in general, are arguably best studied
empirically because their worst-case convergence
properties are often conservative.  These experiments
are designed to shed light on two key questions:
\begin{compactenum}
\item How do these iterative algorithms converge to
the exact solution?
\item Are the techniques faster than a conjugate
gradient based algorithm?
\end{compactenum}
Note that column-wise commute time measure is
a special case for reasons we discuss below, and we only
investigate the accuracy of our procedure for that problem.

\paragraph{Experimental settings} We implemented our methods
in {\sc{Matlab}} and {\sc{Matlab}} mex codes. All computations and
timings were done in Linux on a desktop with a Core i7-960
processor (4 core, 2.8GHz) with 24GB of memory.
As mentioned in the introduction, all of the experimental
code is available from
\url{http://cs.purdue.edu/homes/dgleich/publications/2011/codes/fast-katz/}.

We first describe the data
used in the experiments.  These data were also used
in the experiment about localization in the Katz scores
from the previous section.

%Evaluate the algorithms empirically to ascertain their real-world performance.  Let us outline the experiments:
%
%0) Data:
%
%Graph, time required for 200 svd components
%
%i) Pairwise algorithms:
%---- convergence behavior for commute time
%---- convergence behavior for Katz scores
%---- relevative performance compared with conjugate gradient
%---- sample runtimes
%
%ii) Column-wise algorithms:
%---- convergence behavior of the "top-k" sets
%---- sample runtimes
%---- graph access vs. localization for Katz

\subsection{Data}
\label{sec:data}

We use three publicly
available sources and three graphs we collected ourselves.
The majority of the data comes from the SNAP collection
\cite{Leskovec2010-snap} of which, we use
ca-GrQc, ca-HepTh, ca-CondMat, ca-AstroPh, email-Enron,
email-EuAll~\cite{Leskovec-2007-densification},
wiki-Vote~\cite{Leskovec-2010-signed},
soc-Epinions1~\cite{Richardson-2003-trust},
and soc-Slashdot0811~\cite{Leskovec-2009-community-structure}.
Besides these, the graph tapir
is from \citet{Gilbert-2002-meshpart},
the graph Stanford3 is from \cite{Traud-2011-facebook},
and both graphs stanford-cs~\cite{hirai2000-webbase}
and hollywood-2009 \cite{Boldi-2011-layered} are distributed via the webgraph
framework~\cite{boldi2004-webgraph}.  The graph stanford-cs is
actually a subset of the webbase-2001 graph~\cite{hirai2000-webbase},
restricted only to the pages in the domain cs.stanford.edu.
All graphs are symmetrized (if non-symmetric) and stripped of 
any self-loops, edge weights, and extraneous connected components.

\paragraph{DBLP}
We extracted the DBLP coauthors graph from a recent snapshot (2005-2008) of the DBLP
database. We considered only nodes (authors) that have at least three publications in
the snapshot. There is an undirected edge between two authors if they have coauthored a
paper. From the resulting set of nodes, we randomly chose a sample of 100,000 nodes,
extracted the largest connected component, and discarded any weights on the edges.

\paragraph{arXiv}
This dataset contains another coauthorship graph extracted by a snapshot (1990-2000) of
arXiv, which is an e-print service owned, operated and funded by Cornell University,
and which contains bibliographies in many fields including computer science and
physics. This graph is much denser than DBLP. Again, we extracted the largest connected
component of this graph and only work with that subset.

\paragraph{Flickr contacts}
Flickr is a popular online-community for sharing photos, with millions of users. The node set
represents users, and the directed edges are between contacts.
We start with a crawl extracted from Flickr
in May 2006. This crawl began with a single user and continued until the total
personalized PageRank on the set of uncrawled nodes was less than 0.0001. The result of
the crawl was a graph with 820,878 nodes and 9,837,214 edges. In order to create a sub-
graph suitable for our experimentation we performed the following steps. First, we
created a graph from Flickr by taking all the contact relationships that were
reciprocal, and second, we again took the largest connected component.
(This network is now available from the University of Florida sparse
matrix collection~\cite{Davis2011-matrix}).

Table~\ref{tab:data} presents some elementary statistics about these graphs.
We also include the time to compute the truncated singular value decomposition
for the first 200 singular values and vectors using the
ARPACK library~\cite{lehoucq1997-arpack}
in Matlab's \texttt{svds} routine.
This time reflects the
work it would take to use the standard low-rank preprocessing algorithm for
Katz scores on the network \cite{cikm03}.

\begin{table}
\caption{The networks studied in the experiments.  The first
five columns are self explanatory.  The last two columns show
the largest singular value of the network, which is also
the matrix 2-norm, and the time taken to compute the largest 200
singular values and vectors.}
\label{tab:data}
\small
\begin{tabularx}{\linewidth}{lXXXXXX}
\toprule
Graph & Nodes & Edges & Avg.~Deg. & Max~Deg. & $\|\mA\|_2$ & \rlap{SVD~(sec.)} \\
\midrule
               tapir &    1024 &     2846 &   5.56 &    24 &    6.9078 & 2.2 \\
         stanford-cs &    2759 &    10270 &   7.44 &   303 &   39.8213 & 8.9 \\
             ca-GrQc &    4158 &    13422 &   6.46 &    81 &   45.6166 & 16.2 \\
            ca-HepTh &    8638 &    24806 &   5.74 &    65 &   31.0348 & 31.5 \\
          ca-CondMat &   21363 &    91286 &   8.55 &   279 &   37.8897 & 78.6 \\
           wiki-Vote &    7066 &   100736 &  28.51 &  1065 &  138.1502 & 28.5 \\
            ca-HepPh &   11204 &   117619 &  21.00 &   491 &  244.9349 & 49.5 \\
                dblp &   93156 &   178145 &   3.82 &   260 &   33.6180 & 391.0 \\
         email-Enron &   33696 &   180811 &  10.73 &  1383 &  118.4177 & 119.5 \\
          ca-AstroPh &   17903 &   196972 &  22.00 &   504 &   94.4296 & 62.3 \\
         email-EuAll &  224832 &   339925 &   3.02 &  7636 &  102.5365 & 935.3 \\
       soc-Epinions1 &   75877 &   405739 &  10.69 &  3044 &  184.1751 & 324.6 \\
    soc-Slashdot0811 &   77360 &   469180 &  12.13 &  2539 &  131.3418 & 359.1 \\
               arxiv &   86376 &   517563 &  11.98 &  1253 &   99.3319 & 241.2 \\
           Stanford3 &   11586 &   568309 &  98.10 &  1172 &  212.4606 & 48.8 \\
             flickr2 &  513969 &  3190452 &  12.41 &  4369 &  663.3587 & 3418.7 \\
      hollywood-2009 & 1069126 & 56306653 & 105.33 & 11467 & 2246.5596 & 5998.9 \\
\bottomrule			
\end{tabularx}
\end{table}

\subsection{Pairwise commute scores}
\label{sec:experiments-pairwise-commute}

From this data, we now study the performance of our algorithm for
pairwise commute scores, and compare it against
solving the linear system $\mLhat \vx = (\ve_i - \ve_j)$ using
the conjugate gradient method (CG).  At each step of CG,
we use the approximation $(\ve_i - \ve_j)^T \vx\itn{k}$,
where $\vx\itn{k}$ is the $k$h iterate.
The convergence check in CG
was either the pairwise
element value changed by less than the tolerance, checked by
taking a relative difference between steps, or the 2-norm of
the residual fell below the tolerance.

The first figure we present shows the
result of running Algorithm~\ref{alg:pairwise-commute}
on a single pairwise commute time
problem for few graphs (Figure~\ref{fig:pairwise-commute-conv}).
The upper row of figures show the actual bounds themselves. 
The bottom row of figures shows the relative error 
that would result from using the bounds as an
approximate solution.  We show the same results for CG.  The
exact solution was computed by using MINRES~\cite{Paige-1975-indefinite}
to solve the
same system as CG to a tolerance of $10^{-10}$.
For all of the graphs, we used $\lmin = 10^{-4}$ and
$\lmax=\normof[1]{\mLhat}$.  Again using ARPACK, we verified that
the smallest eigenvalue of each of the Laplacian matrices was larger than
$\lmin$.
We chose the vertices for the pair from among the high-degree
vertices for no particular reason.
Both Algorithm~\ref{alg:pairwise-commute} and CG used a tolerance
of $10^{-4}$.

\begin{figure}
\includegraphics[width=0.245\linewidth]{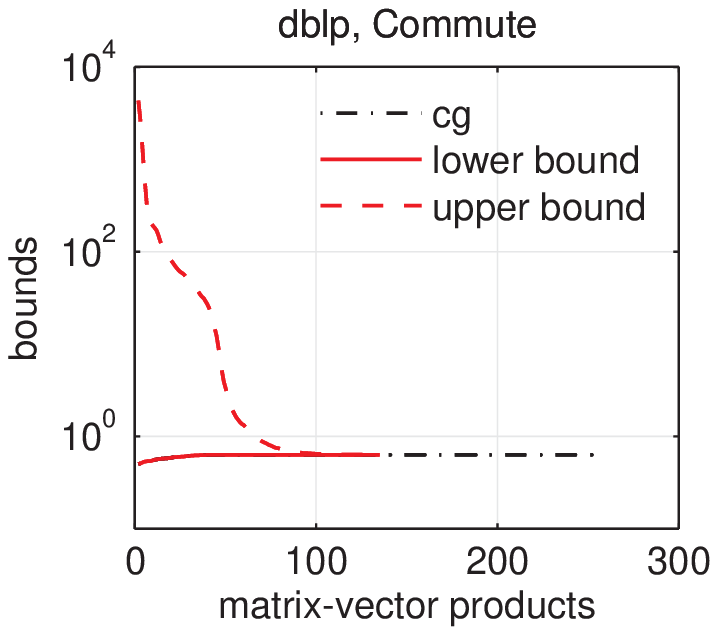}
\includegraphics[width=0.245\linewidth]{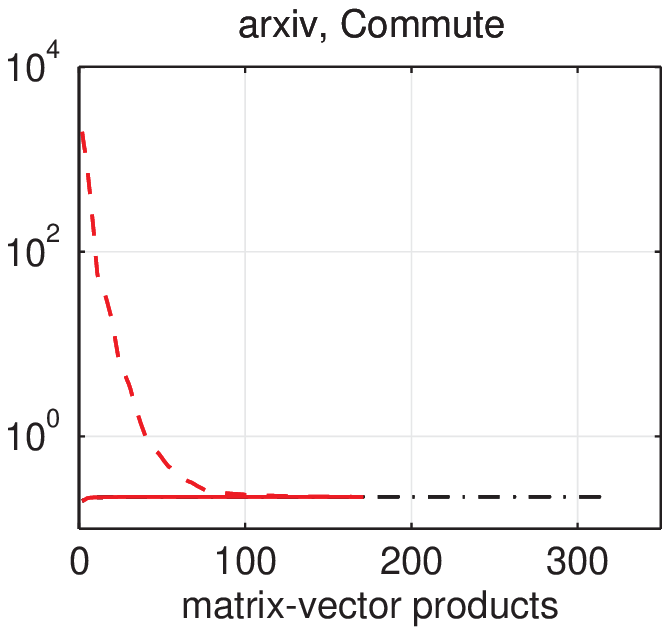}
\includegraphics[width=0.245\linewidth]{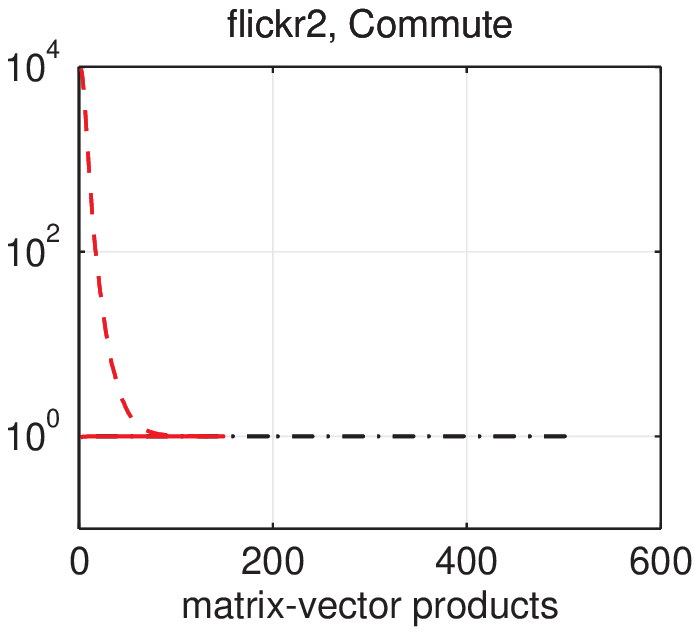}
\includegraphics[width=0.245\linewidth]{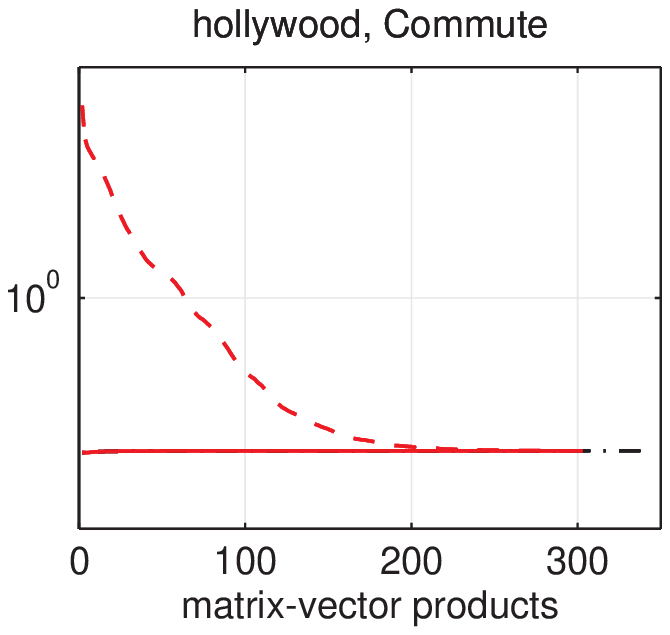}

\includegraphics[width=0.245\linewidth]{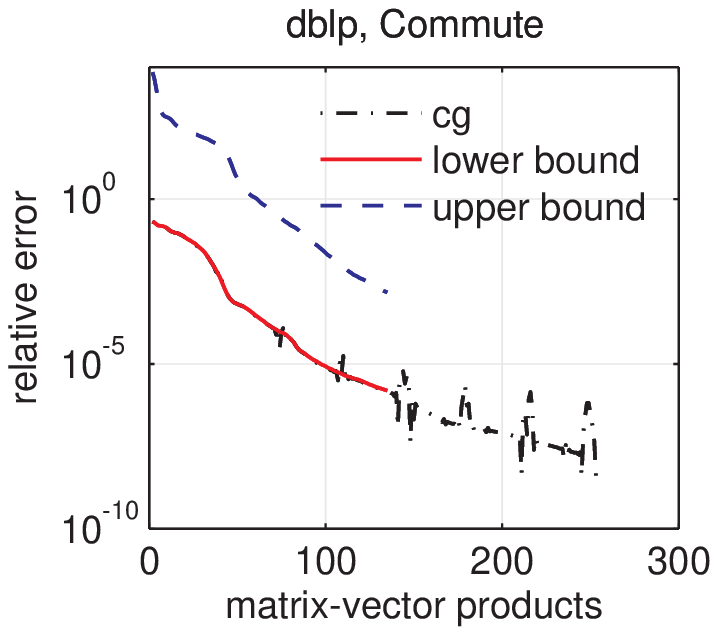}
\includegraphics[width=0.245\linewidth]{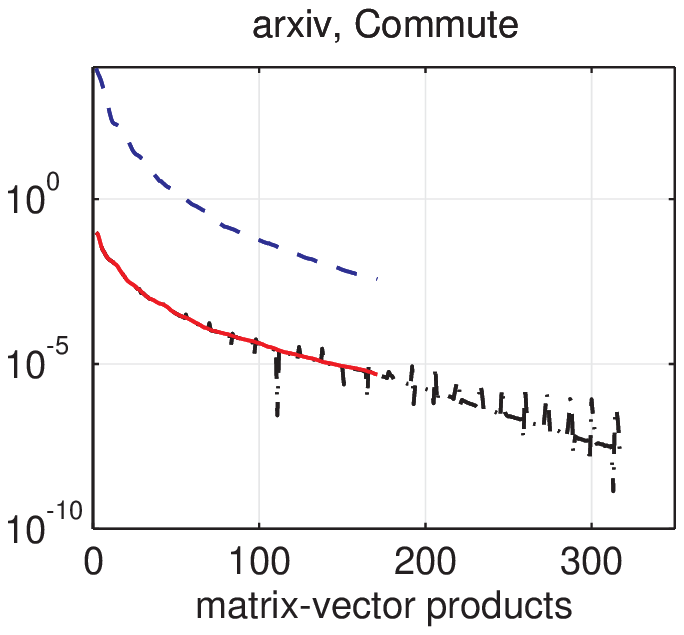}
\includegraphics[width=0.245\linewidth]{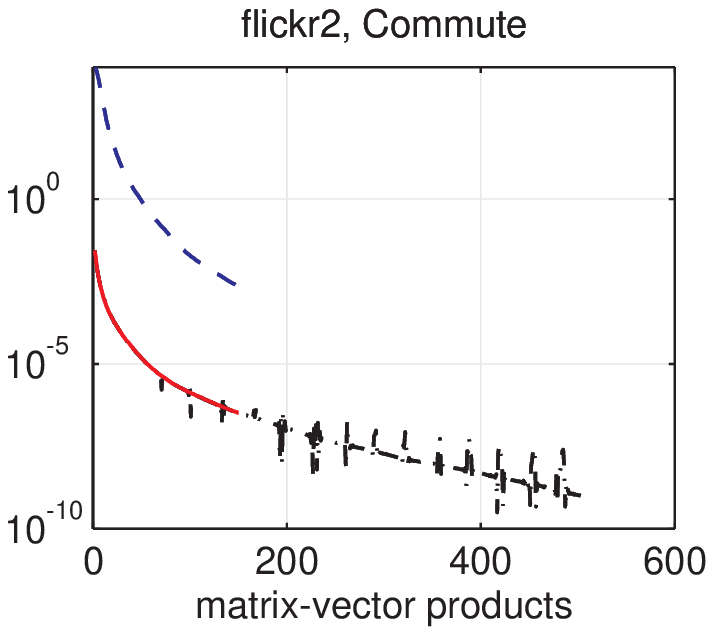}
\includegraphics[width=0.245\linewidth]{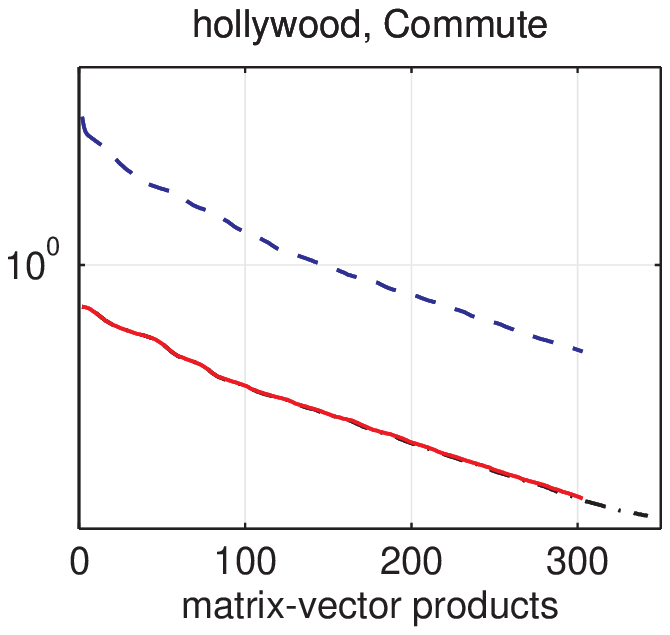}

\caption{Convergence results for pairwise commute times.
(Top row) Each figure shows the upper
and lower bounds at each iteration of Algorithm~\ref{alg:pairwise-commute}
for the graphs dblp, arxiv, flickr2, and hollywood-2009.
(Bottom row)  For the same graphs, each figure shows the
relative size of the error, $(v_{\text{alg}} - v_{\text{exact}})/v_{\text{exact}}$
in the upper and lower bounds at each iteration.
In both cases, we also show the same data from the conjugate gradient
algorithm.  See Section~\ref{sec:experiments-pairwise-commute} for our discussion.}
\label{fig:pairwise-commute-conv}
\end{figure}

In the figure, the upper bounds and lower bounds ``trap'' the solution
from above and below. These bounds converge smoothly to the
final solution. For these experiments, the lower bound has smaller error,
and also, this error tracks the performance of CG quite
closely.  This behavior is expected in cases
where the largest eigenvalue of the matrix is well-separated from the remaining
eigenvalues -- a fact that holds for the Laplacians of our graphs,
see \citet{mihail2002-eigenvalue-power-law} and \citet{Chung-2003-Eigenvalues}
for some theoretical justification.
When this happens, the Lanczos procedure underlying both
our technique and CG quickly produces an accurate estimate of the true
largest eigenvalue, which in turn eliminates any effect due to our initial
overestimate of the largest eigenvalue.  (Recall from Algorithm~\ref{alg:pairwise-commute}
that the estimate of $\lmax$ is present in the computation of the lower-bound
$\underline{b_j}$.)

Here, the conjugate gradient method suffers two problems.  First, because
it does not provide bounds on the score, it is not possible to terminate
it until the residual is small.  Thus,
the conjugate gradient method requires more iterations than
our pairwise algorithm.  Note, however, this result is simply a matter
of detecting when to stop -- both conjugate gradient and our lower-bound
produce similar relative errors for the same work.
Second, the relative error for conjugate gradient displays erratic behavior.
Such behavior is not unexpected because conjugate gradient optimizes
the $A$-norm of the solution error and it is not guaranteed
to provide smooth convergence in true error norm.  These oscillations
make early termination of the CG algorithm problematic, whereas
no such issues occur for the upper and lower bounds from our
algorithm. We speculate that the seemingly smooth convergence behavior that we observe for the upper and lower bound estimates may be rooted in the convergence behavior of the largest Ritz value of the tridiagonal matrix associated with Lanczos, but a better understanding of this issue will require further exploration.

\subsection{Pairwise Katz scores}
\label{sec:experiments-pairwise-katz}
We next show the same type of figure but for
the pairwise Katz scores instead; see
Figure~\ref{fig:pairwise-katz-conv}.
We use a value of $\alpha$
that makes $\mI - \alpha \mA$ nearly indefinite.  Such a value
produces the slowest convergence in our experience.  The
particular value we use is $\alpha = 1/(\|\mA\|_2 + 1)$,
which we call ``hard alpha'' in some of the figure titles.
For all of the graphs, we again used $\lmin = 10^{-4}$ and
$\lmax=\normof[1]{\mLhat}$. This value of $\lmin$ is
smaller than the smallest eigenvalue of $\eye - \alpha \mA$
for all the graphs.
Also, the vertex
pairs are the same as those used for
Figure~\ref{fig:pairwise-commute-conv}.
For pairwise Katz scores, the baseline approach
involves solving the linear system $(\eye - \alpha \mA) \vx = \ve_j$,
again using the conjugate gradient method (CG).  At each step of CG,
we use the approximation $\ve_i^T \vx\itn{k}$,
where $\vx\itn{k}$ is the $k$h iterate.
We use the same convergence check as in the CG baseline for
commute time.
For these problems, we also
evaluated techniques based on the Neumann series for $\eye-\alpha
\mA$, but those took over 100 times as many iterations as CG or
our pairwise approach. The Neumann series is the same algorithm
used in \cite{Wang2007-link-prediction} but customized for the
linear system, not matrix inverse, which is a
more appropriate comparison for the pairwise case.
Finally, the
exact solution was again computed by using MINRES~\cite{Paige-1975-indefinite}
to solve the same system as CG to a tolerance of $10^{-14}$.

\begin{figure}
\includegraphics[width=0.245\linewidth]{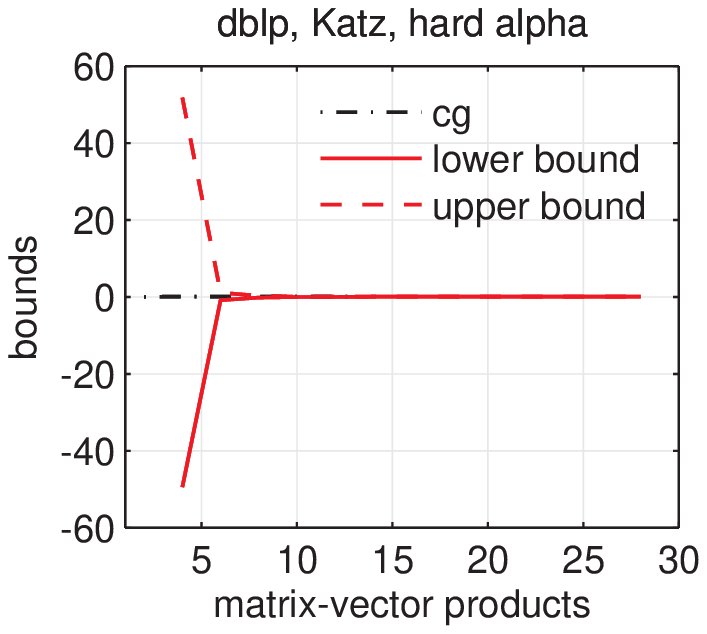}
\includegraphics[width=0.245\linewidth]{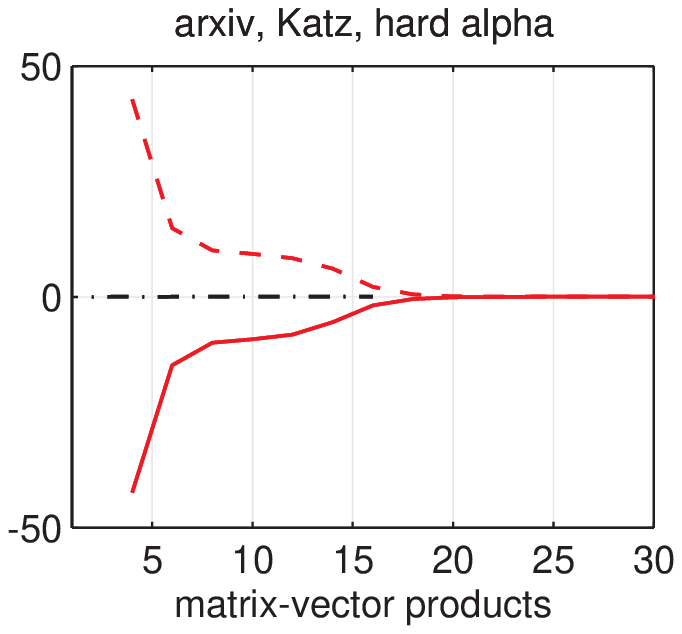}
\includegraphics[width=0.245\linewidth]{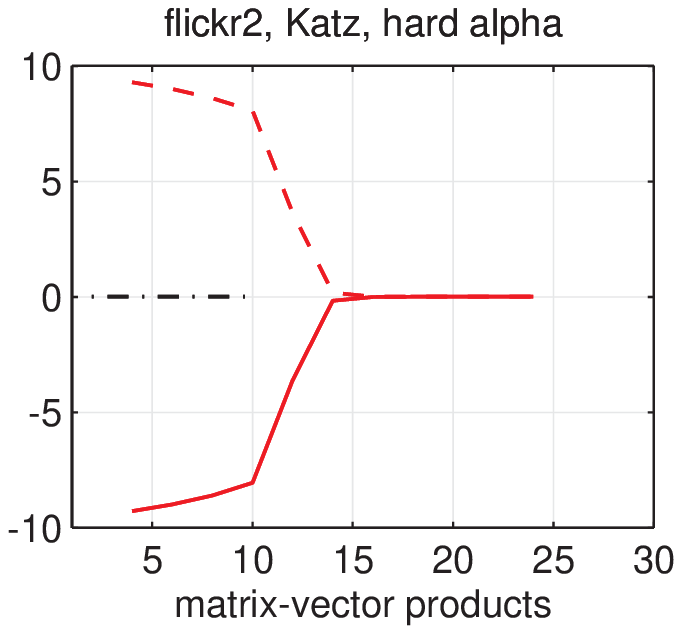}
\includegraphics[width=0.245\linewidth]{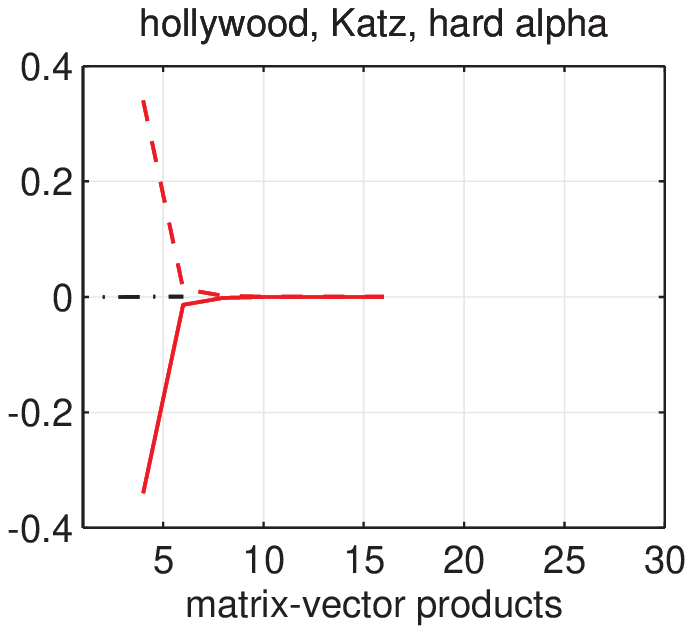}

\includegraphics[width=0.245\linewidth]{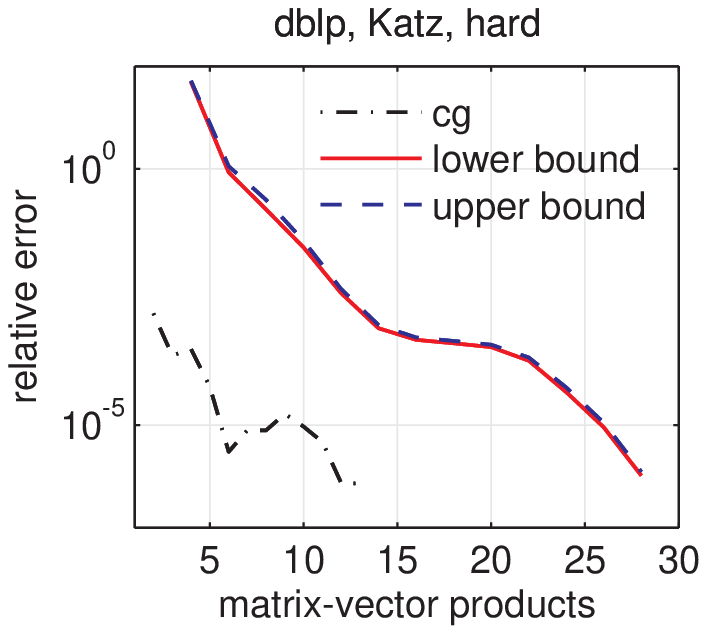}
\includegraphics[width=0.245\linewidth]{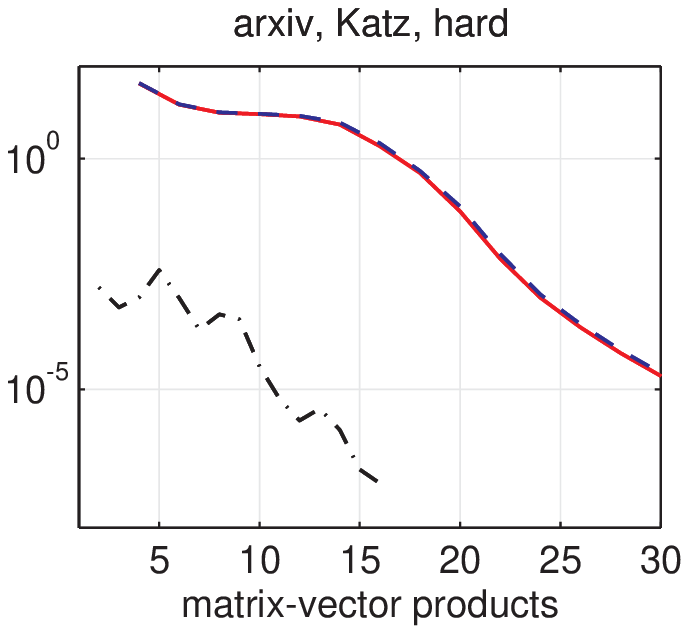}
\includegraphics[width=0.245\linewidth]{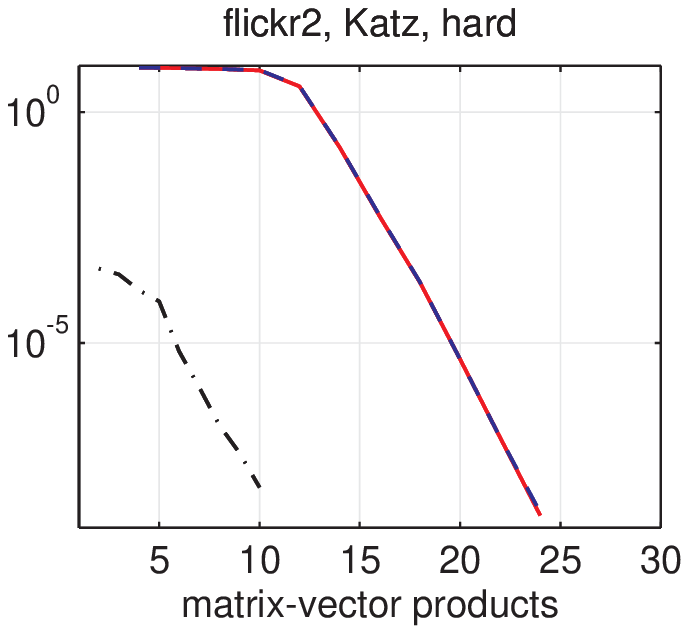}
\includegraphics[width=0.245\linewidth]{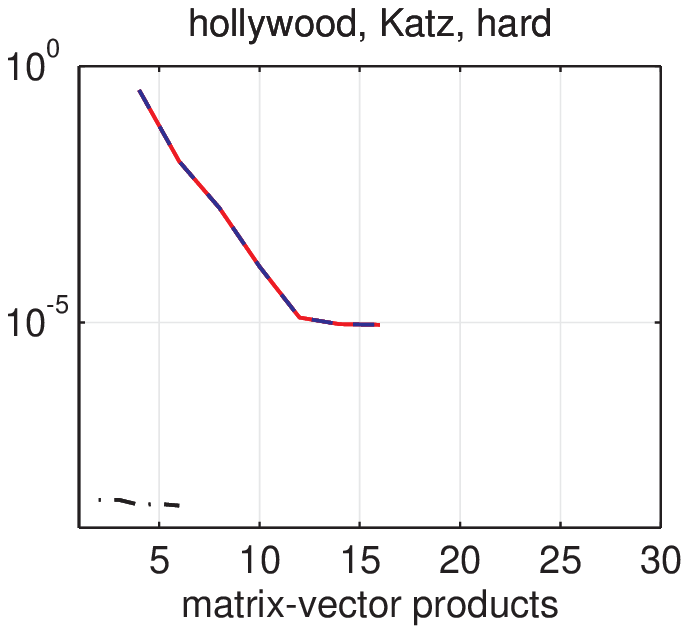}

\caption{Convergence results for pairwise Katz scores.
(Top row.) Each figure shows the upper
and lower bounds at each iteration of Algorithm~\ref{alg:pairwise-katz}
for the graphs dblp, arxiv, flickr2, and hollywood-2009.
(Bottom row.)  For the same graphs, each figure shows the
relative size of the error, $(v_{\text{alg}} - v_{\text{exact}})/v_{\text{exact}}$
in the upper and lower bounds at each iteration.
In both cases, we also show the same data from the conjugate gradient
algorithm.  See Section~\ref{sec:experiments-pairwise-katz} for discussion.}
\label{fig:pairwise-katz-conv}
\end{figure}

A distinct difference from the commute-time results
is that both the lower and upper bounds
converge similarly and have similar error.
This occurs because of the symmetry
in the upper and lower bounds that results from using
the MMQ algorithm twice on the form:
$(1/4)[ (\ve_i + \ve_j)^T (\mI - \alpha \mA)^{-1} (\ve_i + \ve_j)
- (\ve_i - \ve_j)^T (\mI - \alpha \mA)^{-1} (\ve_i - \ve_j) ].$
In comparison with the conjugate gradient method, our pairwise
algorithm is slower to converge.
While the conjugate gradient
method appears to outperform our pairwise algorithms here, recall
that it does not provide any approximation guarantees.
Also, the two matrix-vector product in Algorithm~\ref{alg:pairwise-katz}
can easily be merged into a single ``combined'' matrix-vector
product algorithm.  As we discuss further in the conclusion, such
an implementation would reduce the difference in runtime
between the two methods.

\subsection{Relative matrix-vector products in pairwise algorithms}

Thus far, we have detailed a few experiments describing how the pairwise
algorithms converge.  In these cases, we compared against
the conjugate gradient algorithm for a single pair of vertices
on each graph.  In this experiment, we examine the number of
matrix-vector products that each algorithm requires for a much
larger set of vertex pairs.  Let us first describe how we picked
the vertices for the pairwise comparison.  There were two types of
vertex pairs chosen: purely random, and degree-correlated.  The
purely random choices are simple: pick a random permutation of
the vertex numbers, then use pairs of vertices from this ordering.
The degree correlated pairs were picked by first sorting the
vertices by degree in decreasing order, then picking
the 1st, 2nd, 3rd, 4th, 5th, 10th, 20th, 30th, 40th, 50th,
100th,\ldots vertices from this ordering, and finally, use all
vertex pairs in this subset.  Note that for commute time,
we only used the 1st, 5th, 10th, 50th, 100th,\ldots.
vertices to reduce the total computation time.
For the pairwise commute times, we used 20 random pairs.
and used 100 random pairs for pairwise Katz scores.

In Figure~\ref{fig:pairwise-matvec-perf}, we show the matrix-vector
performance ratio between our pair-wise algorithms and
conjugate gradient.  Let $k_{\text{cg}}$ be the
number of matrix-vector products until CG converges
to a tolerance of $10^{-4}$ (as in previous experiments);
and let $k_{\text{alg}}$ be the number of matrix-vector
products until our algorithm converges.
The performance ratio is
\[
\frac{k_{\text{cg}} - k_{\text{alg}}}{k_{\text{cg}}},
\]
which has a value of $0$ when the two algorithms take
the same number of matrix-vector products, the value 1
when our algorithm takes 0 matrix-vector products,
and the value -1 (or -2) when our algorithm takes twice (or thrice) as
many matrix-vector products as CG.
We display the results as a box-plot of the results from all trials.  There was
no systematic difference in the results between the two types
of vertex pairs (random or degree correlated).

\begin{figure}
\centering
\begin{minipage}{0.48\linewidth}
\includegraphics[width=\linewidth]{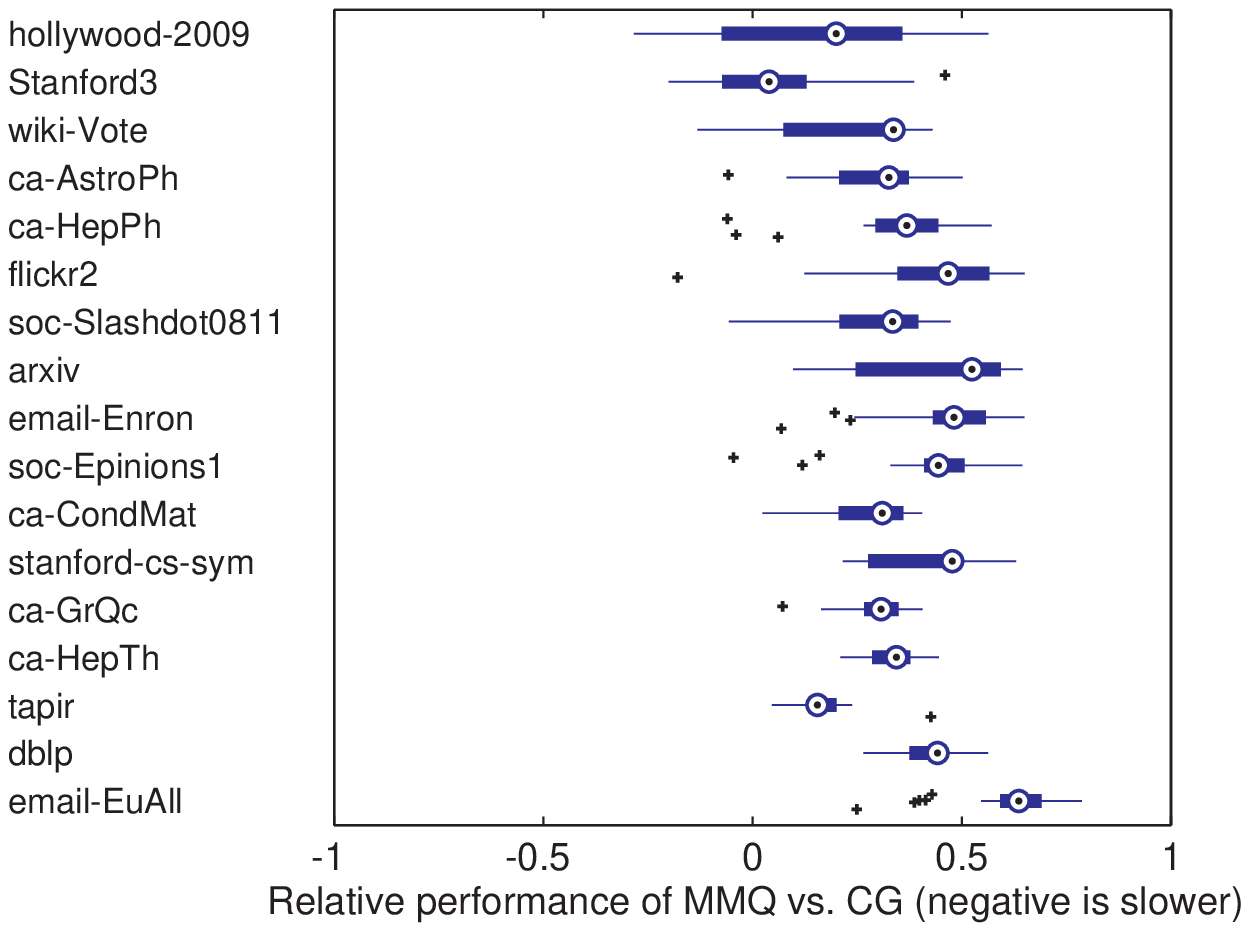}
\end{minipage}
\hfil
\begin{minipage}{0.48\linewidth}
\includegraphics[width=\linewidth]{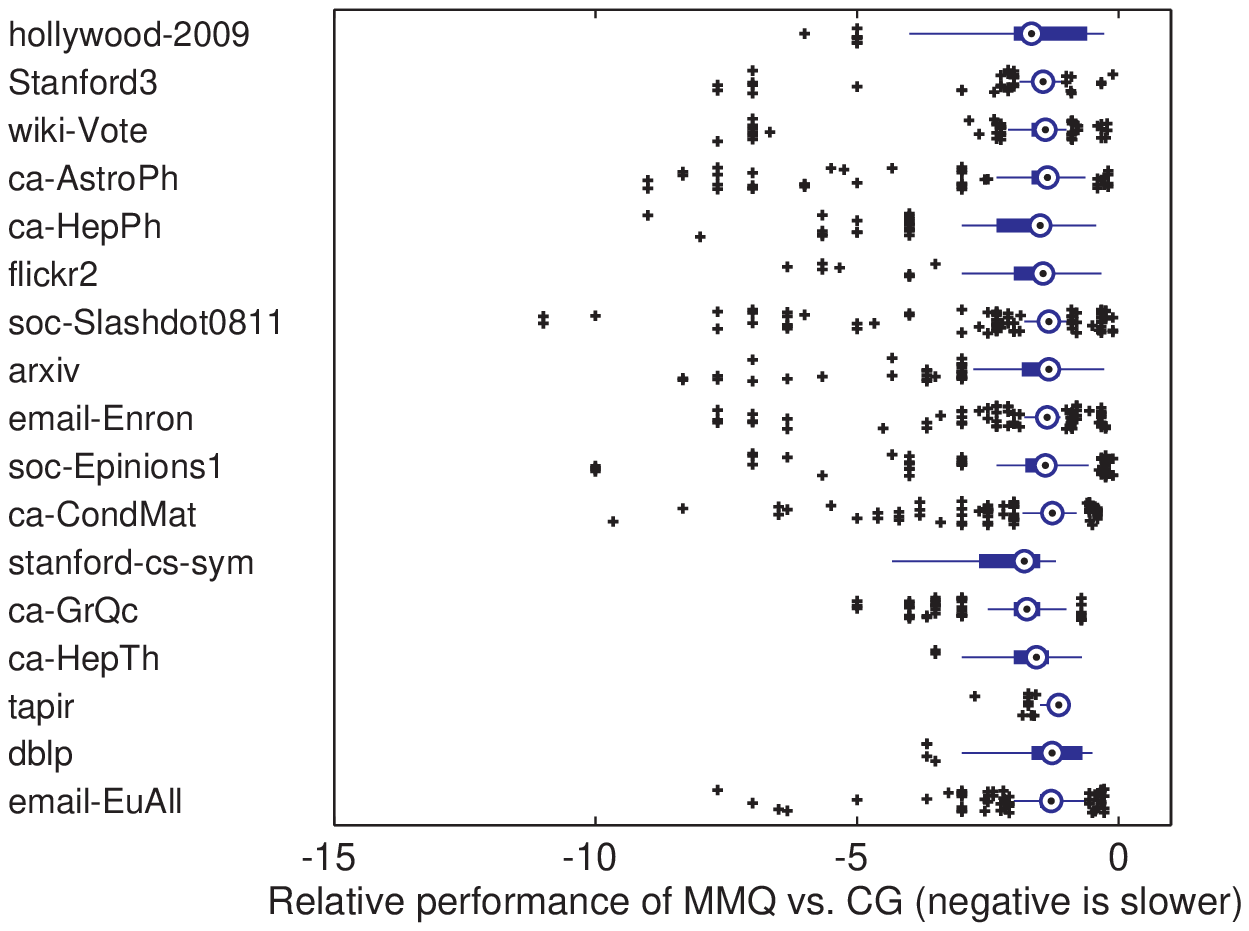}
\end{minipage}
\caption{(Left) Relative performance between 
Algorithm~\ref{alg:pairwise-commute} and 
conjugate gradient for pairwise commute times.
(Right) Relative performance between 
Algorithm~\ref{alg:pairwise-katz} and
conjugate gradient for pairwise Katz scores.
The relative performance measure is
$(k_{\text{cg}} - k_{\text{alg}})/k_{\text{cg}}$,
where $k$ is the number of matrix-vector products
taken by each approach.
}
\label{fig:pairwise-matvec-perf}
\end{figure}

These results show that the small sample in the previous section is
fairly representative of the overall performance difference.
In general, our commute time algorithm uses fewer matrix-vector
products than conjugate gradient.  We suspect this result is
due to the ability to stop early as explained in
Section~\ref{sec:experiments-pairwise-commute}.  And, as also
observed in Section~\ref{sec:experiments-pairwise-katz},
our pairwise Katz algorithm tends to take 2-3 times
as many matrix vector products as conjugate gradient.
These results again used the same ``hard alpha'' value.

\subsection{Column-wise commute times}
\label{sec:experiment-columnwise-commute}

Our next set of results concerns the precision of our
approximation to the column-wise commute time scores.
Because the output of our column-wise commute time algorithm
is based on a coarse approximation of the diagonal elements
of the inverse, we do not expect these scores to converge
to their exact values as we increase the work in the
algorithm.  Consequently, we study the results in
terms of the \emph{precision at $k$} measure.
Recall that the motivation for studying these column-wise measures
is not to get the column scores precisely correct, but
rather to identify the closest nodes to a given query
or target node.  That is, we are most interested in
the smallest elements of a column of the commute time
matrix.  Given a target node $i$, let $S_{k}^{\text{alg}}$ be the $k$
closest nodes to $i$ in terms of our algorithm.  Also, let
$S_{k}^{*}$ be the $k$ closest nodes to $i$ in terms of
the exact commute time.  (See below for how we compute this set.)
The precision at $k$ measure is
\[ | S_{k}^* \cap S_k^{\text{alg}} | / k. \]
In words, this formula computes the fraction of the true
set of $k$ nodes that our algorithm identifies.

We ran the algorithm from Section~\ref{sec:columnwise-commute}
with a tolerance of $10^{-16}$ to evaluate the maximum accuracy
possible with this approach.  We choose 50 target nodes randomly
and also based on the same degree sequence sampling mentioned
in the last section.  For values of $k$ between 5 and 100,
we show a box-plot of the precision at $k$ scores for four networks
in Figure~\ref{fig:commute-columnwise-precision}.
In the same figure, we also show the result of using the heuristic
$C_{i,j} \approx \frac{1}{D_{i,i}} + \frac{1}{D_{j,j}}$ suggested
by \citet{vonLuxburg-2010-commute}.  This heuristic is called
``inverse degree'' in the figure, because it shows that the
set $S_k^{*}$ should look like the set of $k$ nodes with
highest degree or smallest inverse degree.

\begin{figure}

\includegraphics[width=0.245\linewidth]{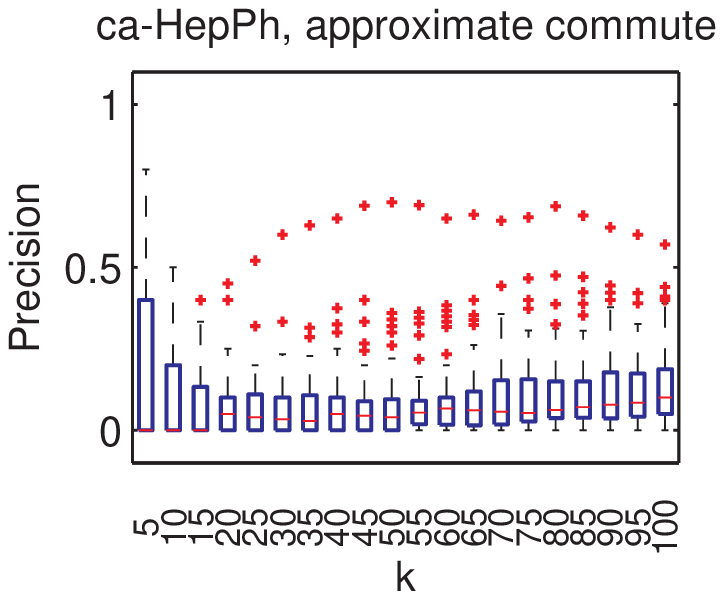}
\includegraphics[width=0.245\linewidth]{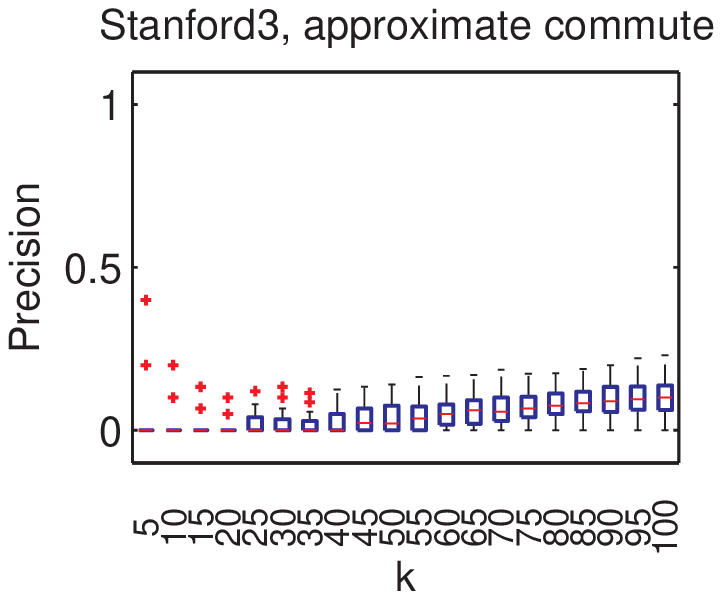}
\includegraphics[width=0.245\linewidth]{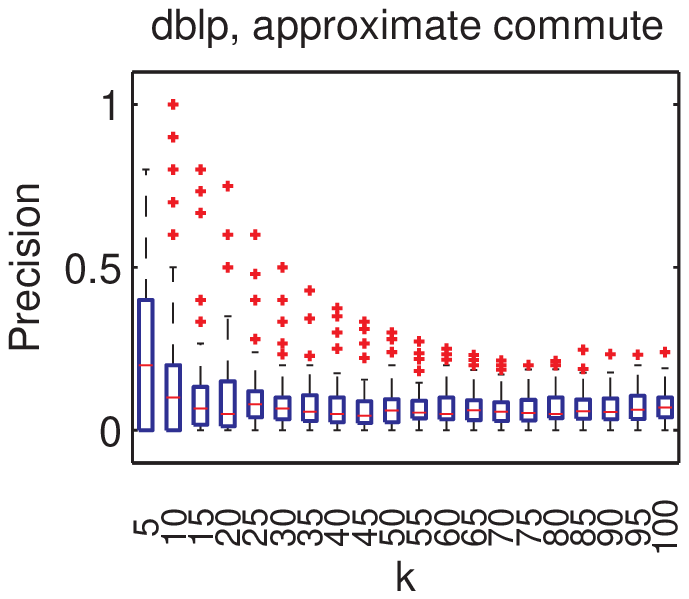}
\includegraphics[width=0.245\linewidth]{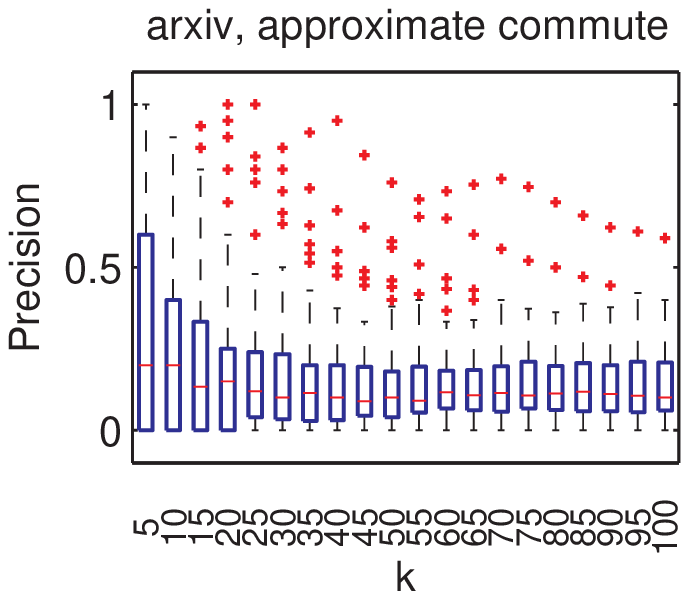}

\includegraphics[width=0.245\linewidth]{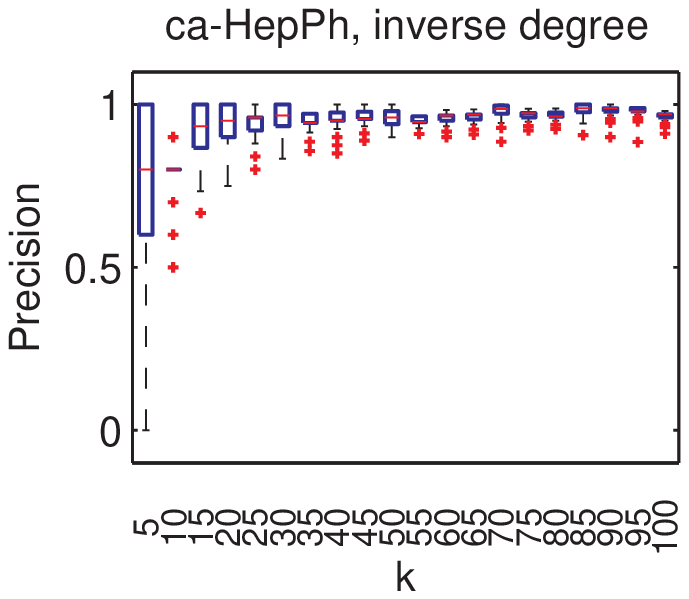}
\includegraphics[width=0.245\linewidth]{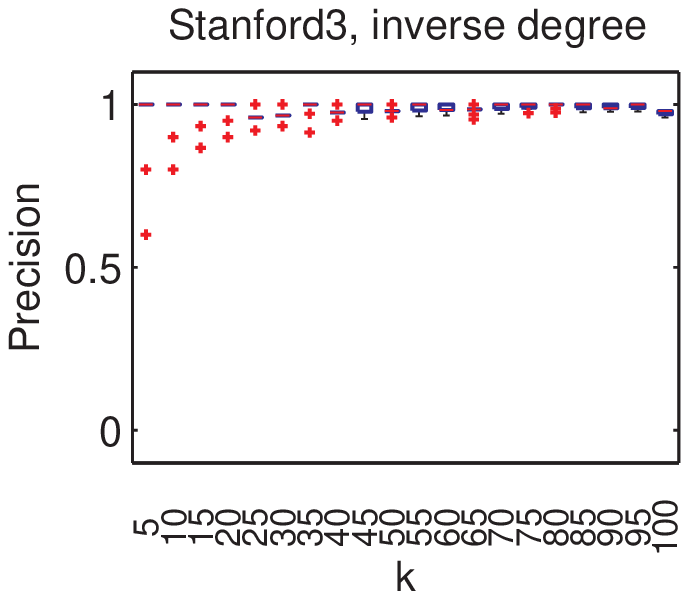}
\includegraphics[width=0.245\linewidth]{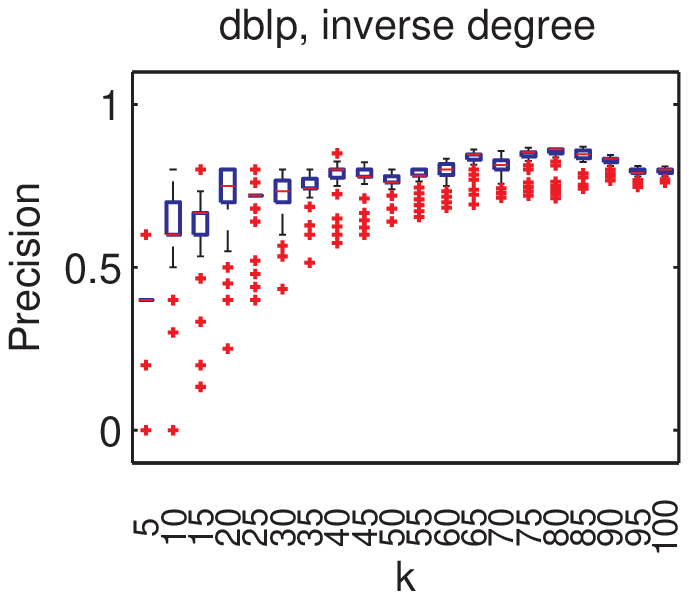}
\includegraphics[width=0.245\linewidth]{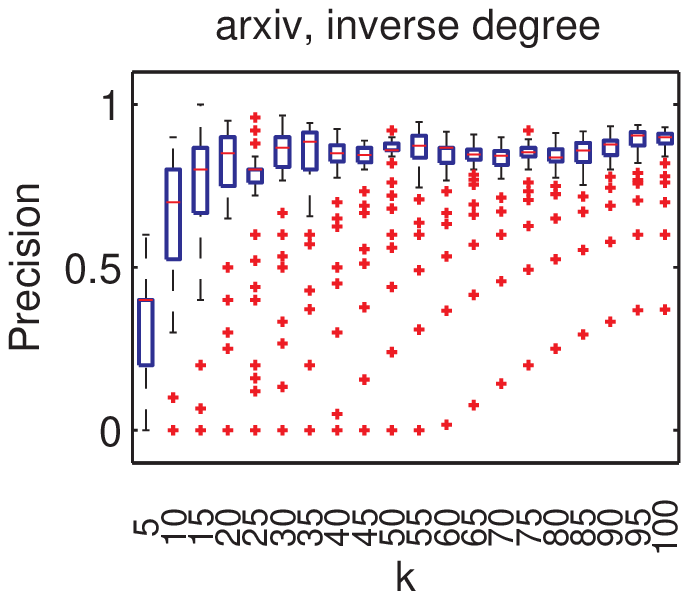}

\caption{Precision at $k$ for the column-wise commute time approximations (top)
over a few hundred trial columns.
Precision at $k$ for the inverse degree heuristic (bottom) over the
same columns.  These figures show standard box-plots of the
result for each column.}
\label{fig:commute-columnwise-precision}
\end{figure}

These results show that our approach for estimating a column
of the commute time matrix provides only partial
information about the true set.  However, these experiments
reinforce the theoretical discussion in \citet{vonLuxburg-2010-commute}
that commute time provides little information beyond the
degree distribution.  Consequently, the results from
our algorithm may provide more useful information in practice.
Although such a conclusion would require us to formalize
the nature of the approximation error in this algorithm,
and involve a rather different kind of study.

\paragraph{Exact commute times}
Computing commute times is challenging.
As part of a separate project, the third author of this
paper wrote a program to compute the exact eigenvalue
decomposition of a combinatorial graph Laplacian in
a distributed computing environment using the MPI and
the ScaLAPACK library~\cite{Blackford-1996-scalapack}.
This program
ignores the sparsity in the matrix and treats the problem
as a dense matrix.  We adapted this
software to compute the pseudo-inverse of the graph Laplacian
as well as the commute times.  We were able to run this
code on graphs up to 100,000 nodes using approximately 10-20
nodes of a larger supercomputer.  (The details varied
by graph, and are not relevant for this paper.)
For graphs with less than 20,000 nodes, the same program
will compute all commute-times on the previously mentioned
desktop computer.  Thus, we computed the exact commute
times for all graphs except email-euAll, flickr2, and
hollywood-2009.

\subsection{Column-wise Katz scores}
\label{sec:experiments-columnwise-katz}

We now come to evaluate the local algorithm for Katz
scores.  As with the pairwise algorithms, we first
study the empirical convergence of the algorithm.
However, the evaluation for the convergence here is
rather different.  Recall, again, that the point of
the column-wise algorithms is to find the most closely
related nodes.  For Katz scores, these are the largest
elements in a column (whereas for commute time,
they were the smallest elements in the column).
Thus, we again evaluate each algorithm in terms of the
precision at $k$ for the top-$k$ set generated by our algorithms
and the exact top-$k$ set produced by solving the linear system.
Natural alternatives are other iterative methods and
specialized direct methods that exploit sparsity. The
latter -- including approaches such as truncated
commute time~\cite{sarkar-moore07} -- are beyond
the scope of this work, since they require a
different computational treatment in terms of caching and
parallelization.  Thus, we again use
conjugate gradient (CG) as our point of comparison.
The exact solution is computed by solving
$(\eye - \alpha \mA) \vk_i = \ve_i$, again using the MINRES
method, to a tolerance of $10^{-12}$.

We also look at the Kendall-$\tau$
correlation coefficient between our algorithm's results and the exact top-$k$
set.  This experiment will let us evaluate whether the algorithm
is ordering the true set of top-$k$ results correctly.  Let
$x_{k^\ast}^\text{alg}$ be the scores from our algorithm on the
exact top-$k$ set, and let $x_{k^\ast}^\ast$ be the  true top-$k$ scores.
The $\tau$ coefficients are computed between $x_{k^\ast}^\text{alg}$
and $x_{k^\ast}^\ast$.

\begin{figure}
\centering

\includegraphics[height=0.22\textheight]{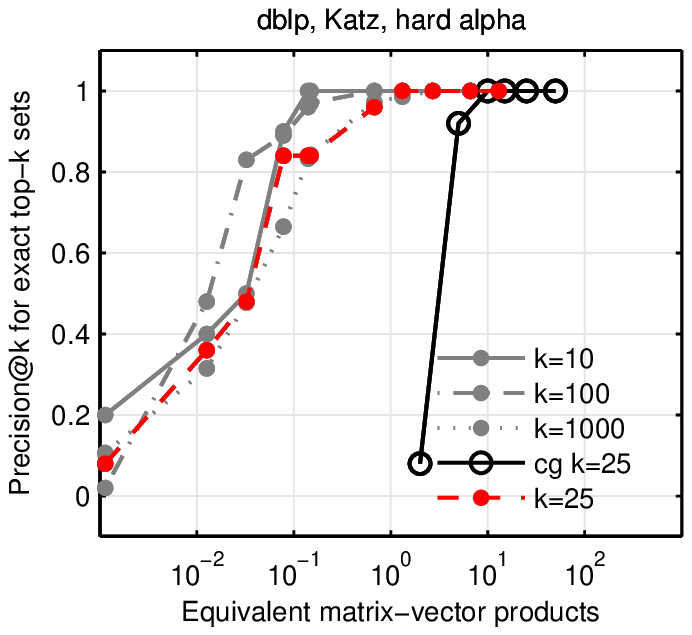}
\includegraphics[height=0.22\textheight]{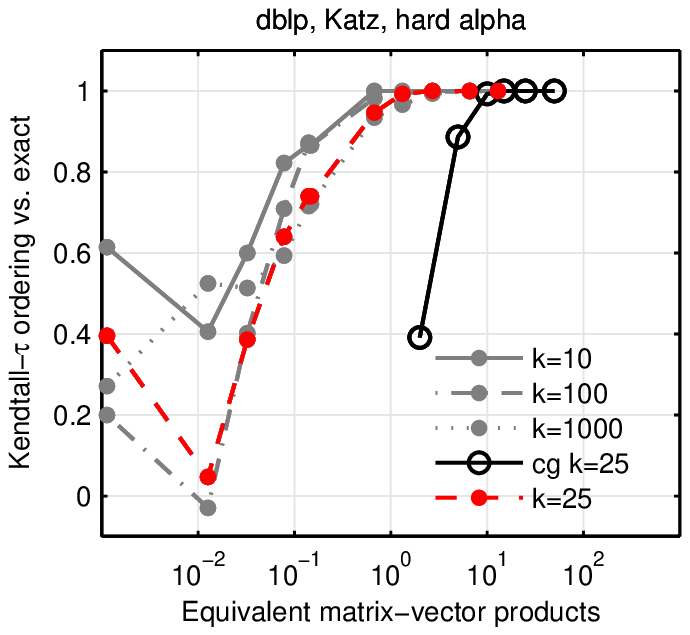}

\includegraphics[height=0.22\textheight]{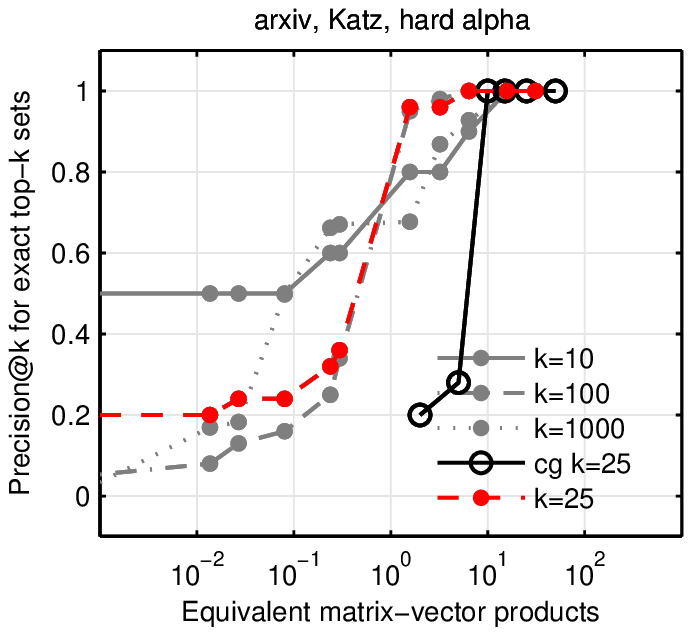}
\includegraphics[height=0.22\textheight]{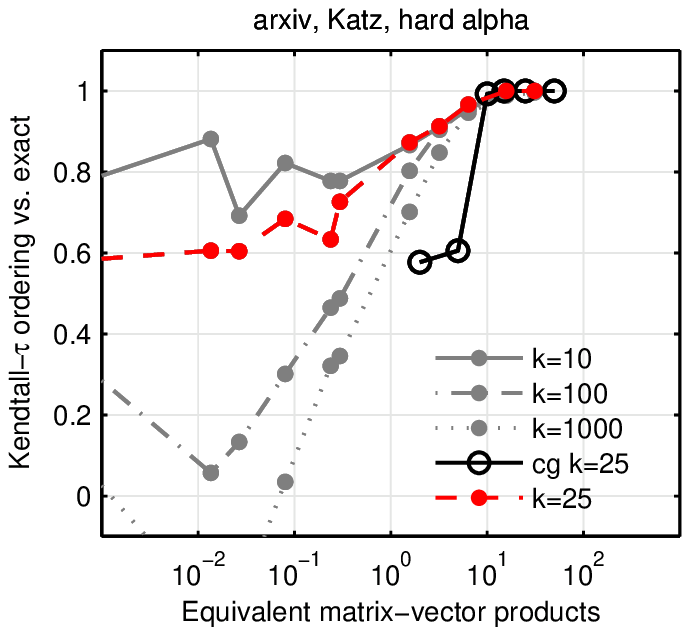}

\includegraphics[height=0.22\textheight]{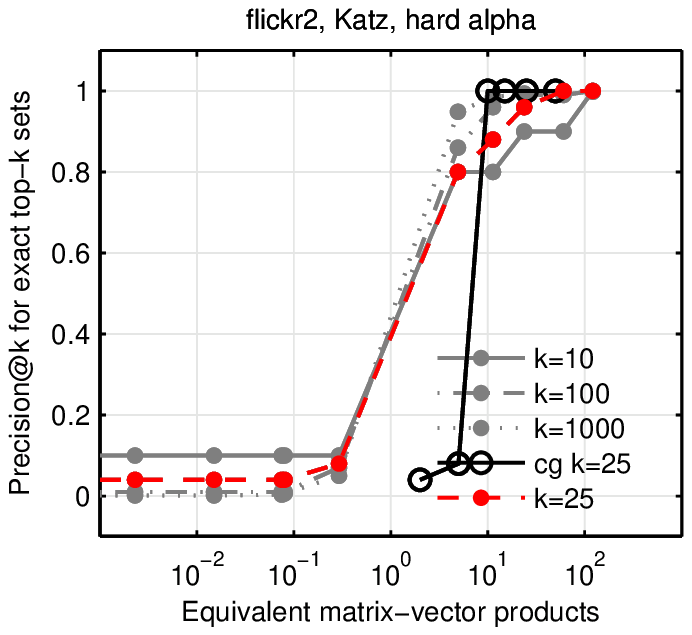}
\includegraphics[height=0.22\textheight]{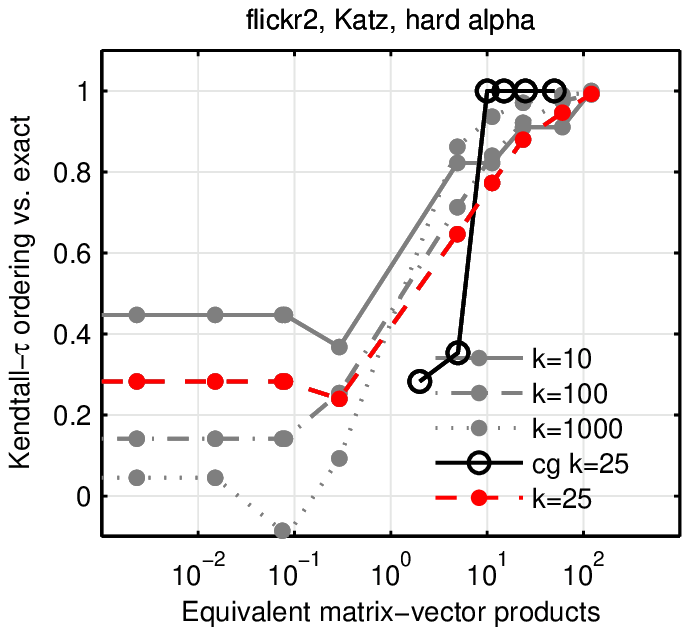}

\includegraphics[height=0.22\textheight]{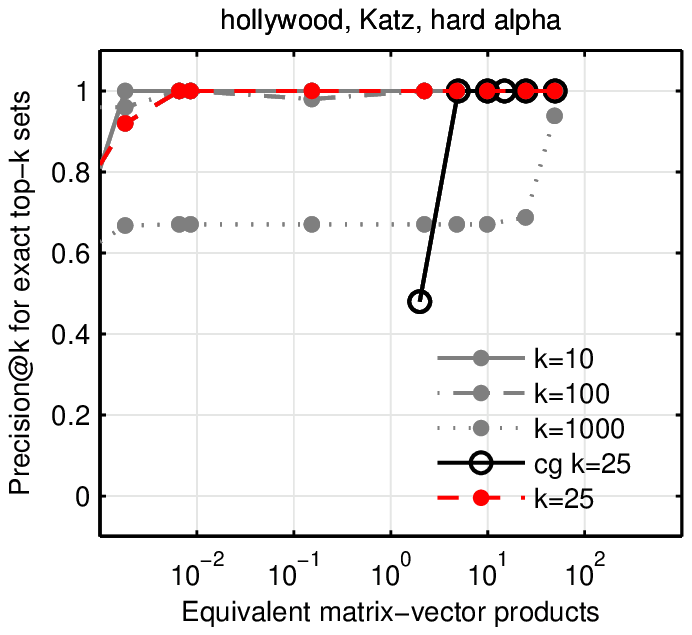}
\includegraphics[height=0.22\textheight]{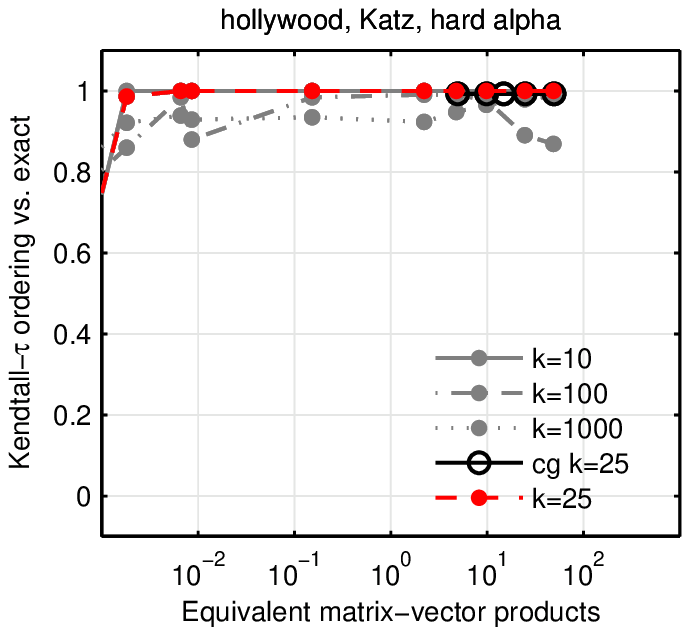}
\caption{Convergence results for our column-wise Katz algorithm in
terms of the precision of the top-$k$ set (left) and the ordering
of the true top-$k$ set (right).  See
Section~\ref{sec:experiments-columnwise-katz} for the discussion.}
\label{fig:katz-columnwise-conv}
\end{figure}

Both of the precision at $k$ and the Kendall-$\tau$
measures should tend to $1$ as we increase the work
in our algorithm.  Indeed, this is what we observe
in Figure~\ref{fig:katz-columnwise-conv}.
For these figures, we pick a vertex with a fairly large
degree and run Algorithm~\ref{alg:katz-columnwise}
with the ``hard alpha'' value mentioned in previous sections.
As the algorithm runs,
we track work with respect to the number of effective
matrix vector products.  An effective matrix-vector product
corresponds to our algorithm examining the same number of edges as a matrix-vector
product.  For example, suppose the algorithm accesses a total of $80$ neighbors
in a graph with $16$ edges.
Then this instance corresponds to $2.5$ effective matrix vector products.
The idea is that the amount of work in one effective matrix vector product
is about the same as the amount of work in one iteration of CG.  Hence,
we can compare algorithms on this ground.
As evident from the legend in each figure, we look at precision at $k$
for four values of $k$, $10, 25, 100, 1000$, and also the Kendall-$\tau$
for these same values.
While all of the measures should tend to 1 as we increase work,
some of the exact top-$k$ results contain
tied values.  Our algorithm has trouble capturing precisely tied values
and the effect is that our Kendall-$\tau$ score does not always tend
to $1$ exactly.

For comparison, we show results from the conjugate gradient
method for the top-$25$ set after $2,5,10,15,25,$ and $50$
matrix-vector products. In
these results, the top-$25$ set is nearly converged after the equivalent of
a single matrix-vector product -- equivalent to just one iteration
of the CG algorithm. The CG algorithm does not provide any useful
information until it converges.  Our top-$k$ algorithm produces
useful partial information in much less work.

\subsection{Runtime}
\label{sec:runtime}

Finally, we show the empirical runtime of our implementations
in Tables~\ref{tab:pairwise-runtime}~and~\ref{tab:katz-columnwise-runtime}.
Table~\ref{tab:pairwise-runtime} describes the runtime of
the two pairwise algorithms.
We show the 25th, 50th, and 75th percentiles of
the time taken to compute the results from Figure~\ref{fig:pairwise-matvec-perf}.
%We did not work to optimize our implementation of this approach,
%and thus, these results are presented as a reference for our performance.
%Our implementation is yet to be fully optimized, but these results seem to provide a reliable indication of performance.
Our implementation is not optimized, and so these results indicate
the current real-world performance of the algorithms.  

Table~\ref{tab:katz-columnwise-runtime} describes the runtime
of the column-wise Katz algorithm.  Here, we picked columns of
the matrix to approximate in two ways: (i) randomly, and (ii)
to sample the entire degree distribution.  As in previous experiments,
we took the 1st, 2nd, 3rd, 4th, 5th, 10th, 20th,\ldots vertices
from the set of vertices sorted in order of decreasing degree.
For each column picked in this manner, we ran
Algorithm~\ref{alg:katz-columnwise} and recorded the wall clock
time.  The 25th, 50th, and 75th percentiles
of these times are shown in the table for each of the two sets of
vertices.

For this algorithm, the \emph{degree} of the target node
has a considerable impact on the algorithm runtime.  This effect
is particularly evident in the flick2 data.  The randomly chosen
columns are found almost instantly, whereas the degree sampled
columns take considerably longer.  A potential explanation for
this behavior is that starting at a vertex with a large degree
will dramatically increase the residual at the first step.  If
these new vertices \emph{also} have a large degree, then this effect
will multiply and the residual will rise for a long time before
converging.  Even in the cases where the algorithm took a long
time to converge, it
only explored a small fraction of the graph (usually about 10\%
of the vertices), and so it retained its locality property.
This property suggests that optimizing our implementation
could reduce these runtimes.

\begin{table}[p]
\caption{Runtime of the pair-wise algorithms.  The ``0.0'' second
entries are rounded down for display.  
These are really just less than 0.1 seconds.
The three columns for each type of problem show the 25th, 50th, and
75th percentiles of the wall-clock time to compute the results
in Figure~\ref{fig:pairwise-matvec-perf}.  
}
\label{tab:pairwise-runtime}
\footnotesize
\begin{tabularx}{\linewidth}{lXXXXXXXX}
\toprule
Graph & Verts. & Avg. & \multicolumn{3}{l}{Pair-wise Katz} & \multicolumn{3}{l}{Pair-wise commute} \\
& & Deg. & \multicolumn{3}{l}{runtime (sec.)} & \multicolumn{3}{l}{runtime (sec.)} \\
\cmidrule{4-9}
& & & 25\% & Median & 75\% & 25\% & Median & 75\%\\
\midrule
               tapir &    1024 &   5.6 &    0.0 &    0.0 &    0.0 &    0.0 &    0.0 &    0.1 \\
         stanford-cs &    2759 &   7.4 &    0.0 &    0.0 &    0.0 &    0.1 &    0.2 &    0.2 \\
             ca-GrQc &    4158 &   6.5 &    0.0 &    0.0 &    0.0 &    0.1 &    0.1 &    0.1 \\
           wiki-Vote &    7066 &  28.5 &    0.0 &    0.0 &    0.0 &    0.2 &    0.2 &    0.2 \\
            ca-HepTh &    8638 &   5.7 &    0.0 &    0.0 &    0.0 &    0.1 &    0.2 &    0.2 \\
            ca-HepPh &   11204 &  21.0 &    0.0 &    0.0 &    0.0 &    0.4 &    0.4 &    0.4 \\
           Stanford3 &   11586 &  98.1 &    0.2 &    0.2 &    0.2 &    0.6 &    0.7 &    0.7 \\
          ca-AstroPh &   17903 &  22.0 &    0.1 &    0.1 &    0.1 &    0.5 &    0.5 &    0.7 \\
          ca-CondMat &   21363 &   8.5 &    0.0 &    0.0 &    0.1 &    0.4 &    0.5 &    0.5 \\
         email-Enron &   33696 &  10.7 &    0.1 &    0.1 &    0.1 &    1.1 &    1.2 &    1.3 \\
       soc-Epinions1 &   75877 &  10.7 &    0.2 &    0.2 &    0.2 &    2.8 &    3.2 &    3.7 \\
    soc-Slashdot0811 &   77360 &  12.1 &    0.2 &    0.2 &    0.2 &    2.6 &    2.8 &    3.4 \\
               arxiv &   86376 &  12.0 &    0.2 &    0.3 &    0.3 &    4.8 &    6.0 &    6.5 \\
                dblp &   93156 &   3.8 &    0.1 &    0.1 &    0.2 &    3.0 &    3.2 &    3.4 \\
         email-EuAll &  224832 &   3.0 &    0.3 &    0.4 &    0.4 &   11.2 &   14.2 &   17.2 \\
             flickr2 &  513969 &  12.4 &    1.3 &    1.7 &    1.8 &   54.8 &   60.0 &   69.8 \\
      hollywood-2009 & 1069126 & 105.3 &   16.5 &   17.0 &   17.4 &  199.2 &  246.0 &  272.5 \\
\bottomrule
\end{tabularx}
\end{table}

\begin{table}[p]
\caption{Runtime of the column-wise Katz algorithm.  The ``0.0'' second
entries are rounded down for display.  These are really just less than 0.1 seconds.
The three columns show the 25th, 50th, and
75th percentiles of the wall-clock time of the experiments
described in Section~\ref{sec:runtime}.  
}
\label{tab:katz-columnwise-runtime}
\footnotesize
\begin{tabularx}{\linewidth}{lXXXXXXXX}
\toprule
Graph & Verts. & Avg. & \multicolumn{3}{l}{Random columns} & \multicolumn{3}{l}{Degree columns} \\
& & Deg. & \multicolumn{3}{l}{runtime (sec.)} & \multicolumn{3}{l}{runtime (sec.)} \\
\cmidrule{4-9}
& & & 25\% & Median & 75\% & 25\% & Median & 75\%\\
\midrule
               tapir &    1024 &   5.6 &    0.0 &    0.0 &    0.0 &    0.0 &    0.0 &    0.0 \\
     stanford-cs-sym &    2759 &   7.4 &    0.0 &    0.0 &    0.0 &    0.0 &    0.0 &    0.0 \\
             ca-GrQc &    4158 &   6.5 &    0.0 &    0.0 &    0.0 &    0.0 &    0.0 &    0.0 \\
           wiki-Vote &    7066 &  28.5 &    0.0 &    0.0 &    0.4 &    0.4 &    0.4 &    0.4 \\
            ca-HepTh &    8638 &   5.7 &    0.0 &    0.0 &    0.0 &    0.0 &    0.0 &    0.0 \\
            ca-HepPh &   11204 &  21.0 &    0.0 &    0.0 &    0.0 &    1.1 &    1.1 &    1.1 \\
           Stanford3 &   11586 &  98.1 &    0.0 &    0.0 &    1.7 &    1.8 &    1.9 &    1.9 \\
          ca-AstroPh &   17903 &  22.0 &    0.0 &    0.0 &    0.0 &    0.6 &    0.7 &    0.9 \\
          ca-CondMat &   21363 &   8.5 &    0.0 &    0.0 &    0.0 &    0.1 &    0.1 &    0.1 \\
         email-Enron &   33696 &  10.7 &    0.0 &    0.0 &    0.0 &    0.9 &    1.0 &    1.1 \\
       soc-Epinions1 &   75877 &  10.7 &    0.0 &    0.0 &    0.0 &    3.7 &    4.1 &    4.5 \\
    soc-Slashdot0811 &   77360 &  12.1 &    0.0 &    0.0 &    0.0 &    2.4 &    2.8 &    3.7 \\
               arxiv &   86376 &  12.0 &    0.0 &    0.0 &    0.0 &    0.0 &    0.6 &    0.7 \\
                dblp &   93156 &   3.8 &    0.0 &    0.0 &    0.0 &    0.0 &    0.0 &    0.0 \\
         email-EuAll &  224832 &   3.0 &    0.0 &    0.0 &    0.0 &    1.1 &    1.7 &    2.5 \\
             flickr2 &  513969 &  12.4 &    0.0 &    0.0 &    0.0 &   11.5 &   52.6 &   55.5 \\
      hollywood-2009 & 1069126 & 105.3 &    0.0 &    0.0 &    0.0 &    0.3 &    0.4 &    0.4 \\
\bottomrule
\end{tabularx}
\end{table}			

%% file: sec-conclusion.tex
\section{Conclusion and Discussion} \label{sec:conclude}

The goal of this manuscript is to estimate commute
times and Katz scores in a rapid fashion.
Let us summarize our contributions and experimental findings.
\begin{compactitem}
\item For the pair-wise commute time problem, we have implemented
Algorithm~\ref{alg:pairwise-commute}, based
on the relationship between the Lanczos process
and a quadrature rule (Section~\ref{sec:mmq}).
This algorithm uses a similar mechanism to that of 
conjugate gradient (CG).  It outperforms 
the latter in terms of total matrix-vector products, 
because it provides upper and lower bounds that allow for 
early termination, whereas CG does not provide an easy way of detecting convergence for a specific pairwise score.
\item For the pair-wize Katz problem, we have proposed
Algorithm~\ref{alg:pairwise-katz}, also based on the same
quadrature theory. This algorithm involves two simultaneous
Lanczos iterations.  In practice, this means 
more work per iteration than a simple approach based on
CG. 
A careful implementation of Algorithm~\ref{alg:pairwise-katz}
would merge the two Lanczos processes into a ``joint process''
and perform the matrix-vector products simultaneously.  In our
tests of this idea, we have found that the combined matrix-vector
product took only 1.5 times as long as a single matrix-vector
product.

\item For the column-wise commute time problem, we have investigated
a variation of the conjugate gradient method
that also provides an estimate of the diagonals of
the matrix inverse.  We have found that these estimates were
fairly crude approximations
of the commute time scores.  We have also investigated
whether the degree-based heuristic from \citet{vonLuxburg-2010-commute}
provides better information.  It indeed seems to perform better, which suggests
that the smallest elements of a column of the commute-time
matrix may not be a useful set of useful related nodes.
\item For the column-wise Katz algorithm, we have proposed
Algorithm~\ref{alg:katz-columnwise} based on
the techniques used for personalized PageRank computing.
The idea with these techniques is to exploit sparsity
in the solution vector itself to derive faster algorithms.
We have shown that this algorithm converged in two cases:
Remark~\ref{rem:katz-simple}, where we established
a precise convergence result, and Theorem~\ref{thm:katz-coordinate},
where we only established asymptotic convergence.
\end{compactitem}
We believe that these results paint a useful picture of the
strengths and limitations of our algorithms.
Here are a few possible directions for future work:

\paragraph{Alternatives for pair-wise Katz.}
First, there are alternatives to using the identity
$\vu^T f(\mZ) \vv = (1/4) (\vu + \vv)^T f(\mZ) (\vu + \vv)
- (1/4) (\vu - \vv)^T f(\mZ) (\vu - \vv)$ in the $\vu \not= \vv$
case.  The first alternative is based on the nonsymmetric Lanczos
process~\cite{Golub-2010-mmq}.  This approach still requires
two matrix-vector products per iteration, but it directly
estimates the bilinear form and also provides upper and lower bounds.
A concern with the nonsymmetric Lanczos process is that it is
possible to encounter degeneracies in the recurrence relationships
that stop the process short of convergence.  Another alternative
is based on the block Lanczos process~\cite{Golub-2010-mmq}.
However, this process does not yet offer upper and lower bounds.
 
\paragraph{A theoretical basis for the localization of Katz scores.}
The inspiration for the column-wise Katz algorithm were the highly successful
personalized PageRank algorithms.  The localization of these personalized PageRank
vectors was made precise in a theorem from \citet{andersen2006-local} that related the
personalized PageRank vector to cuts in the graph.  In brief, if there is a good cut
nearby a vertex, then the personalized PageRank vector will be localized on a few
vertices.  An interesting question is whether or not Katz matrices enjoy a similar
property.  We hope to investigate this question in the future.